\documentclass[11pt]{article}

\usepackage{times}  
\usepackage{mathpazo}
\usepackage{amssymb,amsmath,amsthm}
\usepackage[T1]{fontenc}
\usepackage{tikz}
\usetikzlibrary{calc}
\usepackage{hyperref}
\usepackage{enumitem}

\newtheorem{theorem}{Theorem}[section]
\newtheorem{proposition}[theorem]{Proposition}
\newtheorem{definition}[theorem]{Definition}

\newtheorem{lemma}[theorem]{Lemma}

\newtheorem{corollary}[theorem]{Corollary}

\newtheorem{remark}[theorem]{Remark}

\newcommand{\qedsymb}{\hfill{\rule{2mm}{2mm}}}
\renewenvironment{proof}[1][]{\begin{trivlist}
\item[\hspace{\labelsep}{\bf\noindent Proof#1:\/}] }{\qedsymb\end{trivlist}}

\def\R{\mathbb{R}}

\def\N{\mathbb{N}}

\newcommand\Expec[2]{{{\bf E}_{#1}\left[ {#2} \right]}}

\newcommand{\true}{\mathsf{True}}
\newcommand{\false}{\mathsf{False}}

\newcommand{\NP}{\mathsf{NP}}

\newcommand{\coNPpoly}{\mathsf{coNP/poly}}

\newcommand{\YES}{\mathsf{YES}}

\newcommand{\eps}{\epsilon}
\renewcommand{\epsilon}{\varepsilon}

\newcommand{\linspan}{\mathop{\mathrm{span}}}
\newcommand{\Fset}{\mathbb{F}}
\newcommand{\Kset}{\mathbb{K}}

\newcommand{\qSAT}{q\textsc{-SAT}}
\newcommand{\qNAE}{q\textsc{-NAE-SAT}}
\newcommand{\Col}{\textnormal{\textsc{Coloring}}}
\newcommand{\qCol}{\ensuremath{q}\textnormal{\textsc{-Coloring}}}
\newcommand{\HCol}{\ensuremath{H}\textnormal{\textsc{-Coloring}}}
\newcommand{\ListHCol}{\textnormal{\textsc{List-}}\ensuremath{H}\textnormal{\textsc{-Coloring}}}

\newcommand{\dODP}{\textnormal{\textsc{-Ortho-Dim}}}

\begin{document}

\title{{\bf Kernelization for $H$-Coloring}}

\author{
Yael Berkman\thanks{The Academic College of Tel Aviv-Yaffo, Tel Aviv 61083, Israel. Research supported by the Israel Science Foundation (grant No.~1218/20).}
\and
Ishay Haviv\footnotemark[1]
}

\date{}

\maketitle

\begin{abstract}
For a fixed graph $H$, the $\HCol$ problem asks whether a given graph admits an edge-preserving function from its vertex set to that of $H$. A seminal theorem of Hell and Ne\v{s}et\v{r}il asserts that the $\HCol$ problem is $\NP$-hard whenever $H$ is loopless and non-bipartite. A result of Jansen and Pieterse implies that for every graph $H$, the $\HCol$ problem parameterized by the vertex cover number $k$ admits a kernel with $O(k^{\Delta(H)})$ vertices and bit-size bounded by $O(k^{\Delta(H)} \cdot \log k)$, where $\Delta(H)$ denotes the maximum degree in $H$. For the case where $H$ is a complete graph on at least three vertices, this kernel size nearly matches conditional lower bounds established by Jansen and Kratsch and by Jansen and Pieterse.

This paper presents new upper and lower bounds on the kernel size of $\HCol$ problems parameterized by the vertex cover number. The upper bounds arise from two kernelization algorithms. The first is purely combinatorial, and its size is governed by a structural quantity of the graph $H$, called the non-adjacency witness number. As applications, we obtain kernels whose size is bounded by a fixed polynomial for natural classes of graphs $H$ with unbounded maximum degree, such as planar graphs and, more broadly, graphs with bounded degeneracy. More strikingly, we show that for almost every graph $H$, the degree of the polynomial that bounds the size of our combinatorial kernel grows only logarithmically in $\Delta(H)$. Our second kernel leverages linear-algebraic tools and involves the notion of faithful independent representations of graphs. It strengthens the general bound from prior work and, among other applications, yields near-optimal kernels for problems concerning the dimension of orthogonal graph representations over finite fields. We complement our kernelization results with conditional lower bounds, thereby nearly settling the kernel complexity of the problem for various target graphs $H$.
\end{abstract}

\section{Introduction}

A central concept in graph theory is that of {\em graph homomorphisms}, namely, functions from the vertex set of one graph $G$ to the vertex set of another graph $H$ that map adjacent vertices in $G$ to adjacent vertices in $H$. If such a function exists, the graph $G$ is said to be {\em $H$-colorable}. For a fixed graph $H$, the computational $\HCol$ problem asks whether a given input graph is $H$-colorable. One may always assume that the target graph $H$ is a core, i.e., a graph admitting no homomorphism to a proper subgraph, as replacing $H$ with a minimal core subgraph does not change the problem. This class of problems occupies a fundamental place in computational graph theory and has been studied extensively for decades (see, e.g.,~\cite{HellNBook}). A milestone in the area is the dichotomy theorem of Hell and Ne\v{s}et\v{r}il~\cite{HellN90}, which asserts that the problem is solvable in polynomial time whenever $H$ has a loop or is bipartite, and is $\NP$-hard otherwise. In recent years, $\HCol$ problems have received attention from multiple perspectives, such as algorithmic design~\cite{FominHK07,Wahlstrom11,Rzazewski14}, computational lower bounds~\cite{CyganFGKMPS17}, and fine-grained complexity~\cite{OkrasaR21,PiecykR21,GanianHKOS24}.

This paper investigates the family of $\HCol$ problems from the perspective of parameterized complexity, more specifically, from the viewpoint of kernelization. Parameterized complexity is a framework for analyzing decision problems whose instances are equipped with a quantitative parameter, with the goal of determining the effect of the parameter's value on the problem's computational complexity. A central theme in this field, known as kernelization or data reduction, seeks to design efficient preprocessing algorithms that substantially reduce the input size. These algorithms, called {\em kernels}, take an input instance $(x,k)$, where $x$ is the main input and $k$ is the parameter, and transform it in polynomial time into an equivalent instance $(x',k')$ whose bit-size is bounded by $f(k)$ for a computable function $f:\N \rightarrow \N$ depending only on $k$.
The slowest achievable growth rate of such a function $f$ reflects the problem's compressibility limits with respect to the parameter $k$. In the context of graph problems, a widely-studied and ubiquitous parameter is the size of a vertex cover, i.e., a set of vertices that includes at least one endpoint of every edge in the graph. When $\HCol$ is parameterized by the vertex cover number $k$, the input consists of a graph $G$ along with a vertex cover of $G$ of size $k$.

An archetypal problem within the $\HCol$ framework is the $\qCol$ problem for a fixed integer $q$, which has attracted persistent interest in algorithmic and complexity-theoretic research. This problem asks whether the vertices of a given graph can be colored with $q$ colors, so that no two adjacent vertices share a color. It precisely coincides with the $\HCol$ problem where $H$ is the complete graph on $q$ vertices, and is well known to be $\NP$-hard for $q \geq 3$. The kernelization complexity of the $\qCol$ problem parameterized by the vertex cover number $k$ was first comprehensively studied by Jansen and Kratsch in~\cite{JansenK13} (see also~\cite{JansenThesis}). They showed that for every integer $q \geq 3$, the problem admits a kernel producing graphs with $O(k^q)$ vertices which can be encoded in $O(k^q)$ bits. This result was subsequently refined by Jansen and Pieterse~\cite{JansenP19color}, who obtained a kernel with $O(k^{q-1})$ vertices and bit-size $O(k^{q-1} \cdot \log k)$. Remarkably, this bound is nearly optimal, as it was proved in~\cite{JansenK13,JansenP19color} that for every integer $q \geq 3$ and any real $\eps>0$, the problem does not admit a kernel of size $O(k^{q-1-\eps})$ unless $\NP \subseteq \coNPpoly$, a containment that would imply the collapse of the polynomial-time hierarchy to its third level~\cite{Yap83}.

Another captivating case of the $\HCol$ problem arises from a geometric lens on graph theory, namely, through the concept of {\em orthogonal representations} introduced by Lov{\'{a}}sz~\cite{Lovasz79} in 1979. For a positive integer $d$ and a field $\Fset$, a $d$-dimensional orthogonal representation of a graph $G=(V,E)$ over $\Fset$ assigns to each vertex $v \in V$ a non-self-orthogonal vector $x_v \in \Fset^d$, such that for every edge $\{u,v\} \in E$, the vectors $x_u$ and $x_v$ are orthogonal. Here, two vectors $x,y \in \Fset^d$ are called orthogonal if their standard inner product $\langle x,y \rangle = \sum_{i=1}^{d}{x_i y_i}$ equals zero, and a vector $x \in \Fset^d$ is self-orthogonal if $\langle x,x \rangle =0$. The orthogonal representation is termed {\em faithful} if two vertices $u$ and $v$ are adjacent in $G$ if and only if their vectors $x_u$ and $x_v$ are orthogonal. The minimum dimension $d$ for which a graph admits a (faithful) orthogonal representation over $\Fset$ captures an intriguing structural property, and has found applications in combinatorics, theoretical computer science, and information theory (see~\cite[Chapter~10]{LovaszBook} and, e.g.,~\cite{LovaszSS89,CodenottiPR00,CameronMNSW07,GolovnevRW17,Haviv18free,GolovnevH20}). This motivates the study of the $d\dODP_\Fset$ problem, which, for fixed $d$ and $\Fset$, asks whether a given graph admits a $d$-dimensional orthogonal representation over $\Fset$. This problem naturally fits within the $\HCol$ framework by letting $H = H(\Fset, d)$ be the (possibly infinite) graph whose vertices are all non-self-orthogonal vectors in $\Fset^d$, with adjacency defined by orthogonality over $\Fset$. Note that the problem is known to be $\NP$-hard for every integer $d \geq 3$ and every field $\Fset$~\cite{Peeters96}. It was recently shown~\cite{HavivR24} that for every integer $d \geq 3$ and every field $\Fset$, the $d\dODP_\Fset$ problem parameterized by the vertex cover number $k$ admits a kernel with $O(k^d)$ vertices and bit-size $O(k^d)$, whereas for any $\eps >0$, it admits no kernel of size $O(k^{d-1-\eps})$ unless $\NP \subseteq \coNPpoly$. A sharper upper bound was obtained in~\cite{HavivR24} for the real field $\R$, yielding a near-optimal kernel with $O(k^{d-1})$ vertices and bit-size $O(k^{d-1} \cdot \log k)$. For finite fields, however, the precise kernel complexity remained unresolved, with a multiplicative gap of roughly $k$ separating the upper and lower bounds on the kernel size.

The near-optimal kernels for the $\qCol$ problems were extended by Jansen and Pieterse in~\cite{JansenP19color} to the $\HCol$ problem for arbitrary graphs $H$. They showed that for every graph $H$, the $\HCol$ problem parameterized by the vertex cover number $k$ admits a kernel with $O(k^{\Delta(H)})$ vertices and bit-size $O(k^{\Delta(H)} \cdot \log k)$, where $\Delta(H)$ denotes the maximum degree of a vertex in $H$. This result was actually established in a stronger form, with respect to the less restrictive parameter known as the twin cover number. While the bound on the kernel size from~\cite{JansenP19color} is essentially optimal when $H$ is a complete graph, it is natural to ask to what extent the dependence on the maximum degree of $H$ captures the kernel complexity of the problem for general graphs $H$. This question serves as the driving force behind our work, which explores the polynomial behavior of the kernel complexity of $\HCol$ problems parameterized by the vertex cover number.

\subsection{Our Contribution}

The present paper provides upper and lower bounds on the kernel size of $\HCol$ problems parameterized by the vertex cover number. Our upper bounds include two kernelization algorithms: the first is purely combinatorial, and the second leverages linear-algebraic tools. Below, we elaborate on each of them.

\subsubsection*{The Combinatorial Kernel}
Our first kernel is simple and combinatorial, inspired by the strategy developed for coloring problems by Jansen and Kratsch in~\cite{JansenK13}. Consider an instance of the $\HCol$ problem parameterized by the vertex cover number $k$, namely, a graph $G=(V,E)$ and a vertex cover $X \subseteq V$ of $G$ of size $k$. For a fixed integer $q$, the kernel starts with the subgraph of $G$ induced by $X$, and then, for every set $S \subseteq X$ of size at most $q$ whose vertices share a neighbor from $V \setminus X$ in $G$, introduces a new vertex whose neighborhood is exactly $S$. While the number of vertices in the constructed graph is clearly bounded by $O(k^q)$, we prove that the correctness of the kernelization for $\HCol$ is governed by a structural quantity of the graph $H$, which we denote by $q(H)$ and refer to as the {\em non-adjacency witness number} of $H$.

Given a set of vertices in $H$ that have no common neighbor, one may ask whether this can be certified by a succinct witness, namely, a small subset that also has no common neighbor. The non-adjacency witness number $q(H)$ is defined as the smallest size of such a subset, whose existence is guaranteed for every set of vertices in $H$ with no common neighbor (see Definition~\ref{def:q(G)}). This notion has previously been studied in the literature, with particular attention given to graphs $H$ with $q(H) = 2$ (see, e.g.,~\cite{DouradoPS06,GroshausS07,GroshausLS17}). We note that $q(H)$ may also be formulated as the smallest integer $q$ for which the collection of open neighborhoods in $H$ satisfies what is known as the Helly property of order $q$ (see Section~\ref{sec:NAWN} for details). Our analysis reveals that the size of the aforementioned kernel for the $\HCol$ problem parameterized by the vertex cover number is bounded by a polynomial, whose degree is given by the non-adjacency witness number of $H$.

\begin{theorem}\label{thm:Intro_kernel_q}
For every graph $H$, the $\HCol$ problem parameterized by the vertex cover number $k$ admits a kernel with $O(k^{q(H)})$ vertices and bit-size $O(k^{q(H)})$.
\end{theorem}
\noindent
We remark that when applying Theorem~\ref{thm:Intro_kernel_q}, it is sometimes useful to replace $H$ with a core subgraph of $H$. As noted earlier, this substitution does not change the problem, however, it may reduce the non-adjacency witness number and thereby lead to a smaller kernel (see Lemma~\ref{lemma:q(core)}).

In an attempt to exploit Theorem~\ref{thm:Intro_kernel_q} and to compare it with the kernel of $O(k^{\Delta(H)})$ vertices obtained in~\cite{JansenP19color}, we undertake a systematic study of the non-adjacency witness number of graphs. Our analysis illuminates connections to fundamental graph parameters, such as maximum degree, clique number, and degeneracy, as well as structural properties related to the presence of specific subgraphs. By combining these insights with Theorem~\ref{thm:Intro_kernel_q}, we obtain economical kernels for a broad range of graphs $H$. In particular, as demonstrated below, we achieve kernels whose size is bounded by a fixed polynomial for natural classes of graphs $H$ with unbounded maximum degree.

For a fixed integer $d$, consider the class of $d$-degenerate graphs, namely, those graphs in which every subgraph has a vertex of degree at most $d$. We show that for every such graph $H$, its non-adjacency witness number satisfies $q(H) \leq d+1$, so Theorem~\ref{thm:Intro_kernel_q} implies that the corresponding $\HCol$ problem parameterized by the vertex cover number $k$ admits a kernel whose size is bounded by $O(k^{d+1})$. For planar graphs, which are known to be $5$-degenerate, this yields a bound of $O(k^6)$ on the kernel size, and via a refined analysis, we further reduce it to $O(k^4)$. Perhaps most strikingly, we apply Theorem~\ref{thm:Intro_kernel_q} to show that for {\em almost every graph $H$}, there is a kernel for the $\HCol$ problem parameterized by the vertex cover number, where the degree of the polynomial bounding its size grows only logarithmically with the maximum degree of $H$. To this end, we provide a tight estimate for the typical non-adjacency witness number of random graphs (see Theorems~\ref{thm:G(n,1/2)} and~\ref{thm:kernel_random} for precise statements).

\subsubsection*{The Algebraic Kernel}
Our second kernel for the $\HCol$ problem parameterized by the vertex cover number is more algebraic in nature. For certain graphs $H$, it allows us to sharpen the kernel size from Theorem~\ref{thm:Intro_kernel_q} by a multiplicative linear factor. Somewhat surprisingly, the condition that makes this improvement applicable is related to the concept of faithful orthogonal representations of graphs (which, as described earlier, assign a non-self-orthogonal vector to each vertex, such that two vertices are adjacent if and only if their vectors are orthogonal). Concretely, we show that if a graph $H$ admits a faithful $d$-dimensional orthogonal representation over some computationally efficient field, then the $\HCol$ problem parameterized by the vertex cover number $k$ admits a kernel with $O(k^{d-1})$ vertices and bit-size $O(k^{d-1} \cdot \log k)$.

In fact, we obtain this kernel size under a weaker assumption on the graph $H$, using the concept of {\em independent representations} of graphs. This notion refers to an assignment of a $d$-dimensional vector to each vertex of $H$, so that the vector associated with each vertex does not lie in the vector space spanned by the vectors of its neighborhood. Note that every orthogonal representation is, in particular, an independent representation. The minimum dimension of an independent representation of a graph over a given field characterizes a well-studied graph quantity, called {\em minrank}, which plays a pivotal role in the study of fundamental problems in information theory, such as Shannon capacity~\cite{Haemers79}, index coding~\cite{BBJK06,SDLlocal13,Haviv18}, storage capacity~\cite{Mazumdar15}, and hat guessing games~\cite{Riis07}. For our purposes, we introduce a {\em faithful} analogue of this notion, demanding that two vertices $u$ and $v$ in the graph are adjacent if and only if the vector associated with $u$ lies in the linear subspace spanned by the vectors of the neighborhood of $v$ (see Definition~\ref{def:IR}). With this faithful variant in hand, we prove the following theorem (see also Theorem~\ref{thm:kernel_IR}).

\begin{theorem}[Simplified]\label{thm:Intro_kernel_d}
If a graph $H$ has a faithful $d$-dimensional independent representation over either a finite field or the real field $\R$, then the $\HCol$ problem parameterized by the vertex cover number $k$ admits a kernel with $O(k^{d-1})$ vertices and bit-size $O(k^{d-1} \cdot \log k)$.
\end{theorem}

A few remarks are in order here.
First, it is not difficult to show that every graph $H$ has a faithful independent representation of dimension $\Delta(H)+1$ over any sufficiently large field (see Lemma~\ref{lemma:faith_Delta+1}). By Theorem~\ref{thm:Intro_kernel_d}, this implies that for every graph $H$, the $\HCol$ problem parameterized by the vertex cover number $k$ admits a kernel with $O(k^{\Delta(H)})$ vertices and bit-size $O(k^{\Delta(H)} \cdot \log k)$. Therefore, Theorem~\ref{thm:Intro_kernel_d} strengthens the kernel bound implied by~\cite{JansenP19color}. It is worth noting that a result of Maehara and R{\"{o}}dl~\cite{MaeharaR90} asserts that every graph $H$ has a faithful orthogonal representation over $\R$ of dimension $2 \cdot \Delta(H)$, and it is an open question whether this dimension can be reduced to $\Delta(H)+1$ (see, e.g.,~\cite[Chapter~10.4]{LovaszBook}). Thus, for Theorem~\ref{thm:Intro_kernel_d} to fully subsume the kernel bound of~\cite{JansenP19color}, it is somewhat crucial to adopt the more flexible notion of faithful independent representations (rather than orthogonal).

Second, we observe that for every graph $H$ that possesses a faithful $d$-dimensional independent representation over some field, the non-adjacency witness number satisfies $q(H) \leq d$ (see Lemma~\ref{lemma:q_faith}). This shows that Theorem~\ref{thm:Intro_kernel_d} improves the kernel size obtained via Theorem~\ref{thm:Intro_kernel_q} only when the minimum dimension of a faithful independent representation of $H$ matches the non-adjacency witness number of $H$ (in fact, of a minimal core subgraph of $H$). This situation arises, for example, for the aforementioned graphs $H = H(\Fset,d)$, which correspond to the existence of a $d$-dimensional orthogonal representation over a field $\Fset$. As an application of Theorem~\ref{thm:Intro_kernel_d}, we unleash a near-optimal kernel for the $d\dODP_\Fset$ problems over finite fields $\Fset$, thereby settling a question posed in~\cite{HavivR24}.

\begin{theorem}\label{thm:IntroOD}
For every integer $d \geq 3$ and every finite field $\Fset$, the $d\dODP_\Fset$ problem parameterized by the vertex cover number $k$ admits a kernel with $O(k^{d-1})$ vertices and bit-size $O(k^{d-1} \cdot \log k)$.
\end{theorem}

We further demonstrate the applicability of Theorem~\ref{thm:Intro_kernel_d} for the celebrated family of {\em Kneser graphs}. For positive integers $m$ and $r$ with $m \geq 2r$, the Kneser graph $K(m,r)$ has all $r$-subsets of an $m$-element set as vertices, where two are adjacent if they are disjoint. It turns out that the non-adjacency witness number of $K(m,r)$ is $m-2r+2$, coinciding with its chromatic number that was determined in an influential paper of Lov{\'{a}}sz~\cite{LovaszKneser}. By Theorem~\ref{thm:Intro_kernel_q}, this shows that the $K(m,r)$-$\Col$ problem parameterized by the vertex cover number $k$ admits a kernel with $O(k^{m-2r+2})$ vertices and bit-size. To further reduce this kernel size, we prove that $K(m,r)$ admits a faithful $(m-2r+2)$-dimensional independent representation over any sufficiently large field, which yields, via Theorem~\ref{thm:Intro_kernel_d}, a kernel with $O(k^{m-2r+1})$ vertices and bit-size $O(k^{m-2r+1}\cdot \log k)$ (see Theorems~\ref{thm:kernel_K(m,r)} and~\ref{thm:IR_kneser}). In particular, for the $3$-chromatic Kneser graphs $K(2r+1,r)$, which include the iconic Petersen graph $K(5,2)$, the resulting kernel has near-quadratic size (see Figure~\ref{fig:P}).

The proof of Theorem~\ref{thm:Intro_kernel_d} builds on a powerful sparsification technique that was introduced by Jansen and Pieterse in~\cite{JansenP19sparse} and applied in their work~\cite{JansenP19color}. To illustrate this technique, consider an instance of the $\HCol$ problem parameterized by the vertex cover number, consisting of a graph $G=(V,E)$ and a vertex cover $X \subseteq V$ of $G$. The approach begins by introducing variables over some field $\Fset$ that represent an assignment of vertices from $H$ to the vertices in $X$. Then, for every vertex $v \in V \setminus X$, one constructs one or more low-degree polynomials in the variables corresponding to its neighbors in $X$, which evaluate to zero if and only if $v$ can be assigned a compatible vertex from $H$. The crux is that it suffices to retain in the graph $G$ only a subset of the vertices in $V \setminus X$, whose corresponding polynomials linearly span all the others. The degree of the used polynomials determines the dimension of the entire space of polynomials, and consequently, the resulting kernel size. This method was applied in~\cite{JansenP19color}, where the assignment of each vertex of $X$ is encoded by a set of variables representing an indicator vector over the binary field $\Fset_2$.

The way we apply the machinery of~\cite{JansenP19color,JansenP19sparse} to prove Theorem~\ref{thm:Intro_kernel_d} is substantially different. We rely on the existence of a faithful $d$-dimensional independent representation of the graph $H$ over some field $\Fset$, interpreting its vectors as encodings of the vertices of $H$. Accordingly, we associate each vertex in the vertex cover $X$ with a $d$-dimensional variable vector over $\Fset$ to represent its assignment. The kernelization algorithm first applies our combinatorial kernel to ensure that every vertex outside $X$ has degree at most $d$. Then, it aims to reduce the number of vertices outside $X$ that have degree exactly $d$. By the definition of an independent representation, a vertex $v$ outside $X$ can be assigned a compatible vertex from $H$ only if every set of $d$ vectors of its neighbors in $X$ is linearly dependent over $\Fset$. This condition can be naturally captured by requiring the determinant polynomial over the variables corresponding to such $d$ neighbors of $v$ to be zero. Once the vectors of the neighbors of $v$ span a subspace of dimension smaller than $d$, the assignment of a compatible vertex to $v$ is essentially realized through some other vertex outside $X$ of degree at most $d-1$.

However, the degree of the determinant polynomial of a $d \times d$ matrix is $d$. In order to obtain polynomials of degree $d-1$, which in turn yield a kernel with $O(k^{d-1})$ vertices, we use an idea from~\cite{HavivR24,HavivR25} and require all vectors in the faithful independent representation of $H$ to share a fixed value, say $1$, in their first entry. When the field $\Fset$ is sufficiently large relative to the number of vertices in $H$, this can be achieved by applying an invertible linear transformation to the vectors in the representation to ensure that all first entries are nonzero, and then scaling them appropriately. Over a smaller field, though, such a representation may not exist. We overcome this obstacle by showing that if a graph has a faithful $d$-dimensional independent representation over some field $\Fset$, then it also has one in which all vectors have $1$ as their first entry, over any sufficiently large {\em extension} field of $\Fset$ (see Lemma~\ref{lemma:IR_with1}). This allows us to apply the approach described above over such an extension field and obtain a kernel of the desired size.

\subsubsection*{Lower Bounds}

Given our combinatorial and algebraic kernels for the $\HCol$ problem parameterized by the vertex cover number $k$, it is natural to ask how close the obtained kernel sizes come to being optimal. As already discussed, and shown in~\cite{JansenK13,JansenP19color}, when $H$ is the complete graph on $q$ vertices with $q \geq 3$, the problem admits no kernel of size $O(k^{q-1-\eps})$ for any $\eps >0$, unless $\NP \subseteq \coNPpoly$. This firmly settles the kernel complexity for all cases where $H$, or its minimal core subgraph, is a complete graph. Turning to other fundamental graphs, the first non-trivial case that arises is that of odd cycles of length at least five. For every integer $m \geq 2$, it is easy to verify that the cycle $C_{2m+1}$ on $2m+1$ vertices has non-adjacency witness number $q(C_{2m+1}) = 2$ (see Lemma~\ref{lemma:q(cycle)}). Consequently, Theorem~\ref{thm:Intro_kernel_q} yields a kernel for the $C_{2m+1}$-$\Col$ problem parameterized by the vertex cover number $k$ with $O(k^2)$ vertices and bit-size. Furthermore, neither Theorem~\ref{thm:Intro_kernel_d} nor the kernel from~\cite{JansenP19color} improves upon this bound. Another notable case is the family of $3$-chromatic Kneser graphs $K(2r+1,r)$, for which Theorem~\ref{thm:Intro_kernel_d} yields a kernel with $O(k^2)$ vertices and bit-size $O(k^2 \cdot \log k)$. The following theorem indicates that, in both cases, the obtained kernel sizes are nearly optimal.

\begin{theorem}\label{thm:IntroLower}
For each of the following cases, for any real number $\eps >0$, the $\HCol$ problem parameterized by the vertex cover number $k$ does not admit a kernel of size $O(k^{2-\eps})$ unless $\NP \subseteq \coNPpoly$.
\begin{enumerate}
  \item $H = C_{2m+1}$ for an integer $m \geq 1$.
  \item\label{itm:IntroLowerKneser} $H = K(2r+1,r)$ for an integer $r \geq 1$.
\end{enumerate}
\end{theorem}

We prove Theorem~\ref{thm:IntroLower} in a stronger form in several respects. First, we show that the near-quadratic lower bound persists even when the $H$-$\Col$ problem is parameterized by the total number of vertices in the input graph. Since the entire vertex set of a graph trivially forms a vertex cover, this directly implies the statement above. Second, the proof encompasses not only the graphs $H$ stated in the theorem but also extends to various other graphs $H$. Specifically, the result is proved for all non-bipartite core graphs that are {\em projective}, a class that is known to cover nearly all graphs, including prominent and well-studied graph families (see Section~\ref{sec:lowerConcrete} and, e.g.,~\cite{Larose02}). Third, as is standard in such results, the lower bound applies to the compressibility of the $\HCol$ problem into any target problem, not necessarily itself. Our generalized result is stated as follows.

\begin{theorem}\label{thm:IntroLowerGen}
For every non-bipartite projective core graph $H$ and for any real number $\eps >0$, the $\HCol$ problem parameterized by the number of vertices $n$ does not admit a compression of size $O(n^{2-\eps})$ unless $\NP \subseteq \coNPpoly$.
\end{theorem}

The proof of Theorem~\ref{thm:IntroLowerGen} is based on a compression lower bound, due to Chen, Jansen, Okrasa, Pieterse, and Rz{\k a}\.{z}ewski~\cite{ChenJOPR23}, for a variant of the $\HCol$ problem called $\ListHCol$. In this variant, the input consists of a graph $G$ along with a list of allowed vertices of $H$ for each vertex of $G$, and the goal is to decide whether there exists a homomorphism from $G$ to $H$ that respects these lists (see Definition~\ref{def:list}). This problem generalizes the standard $\HCol$ problem, which corresponds to the special case where all lists are equal to the full vertex set of $H$. Feder, Hell, and Huang~\cite{FederHH99} introduced a class of geometric intersection graphs, called {\em bi-arc graphs}, and used this class to fully characterize the complexity of the $\ListHCol$ problem. They showed that the problem is solvable in polynomial time if $H$ is bi-arc, and $\NP$-hard otherwise. Chen et al.~\cite{ChenJOPR23} proved that, unless $\NP \subseteq \coNPpoly$, the $\ListHCol$ problem parameterized by the number of vertices admits no sub-quadratic compression for any $H$ that is not bi-arc (see Theorem~\ref{thm:ChenJOPR}). They also posed the question of whether such a lower bound holds for the (non-list) $\HCol$ problem, when $H$ is loopless and non-bipartite. As a step toward this challenge, Theorem~\ref{thm:IntroLowerGen} confirms the lower bound for the substantial class of non-bipartite projective core graphs $H$. To this end, we show that for such graphs $H$, instances of $\ListHCol$ can be efficiently reduced into instances of $\HCol$, with the number of vertices preserved up to a multiplicative constant. The proof borrows a gadget construction due to Okrasa and Rz{\k a}\.{z}ewski~\cite{OkrasaR21} and echoes early $\NP$-hardness proofs for $\HCol$ problems~\cite{Nesetril81,MaurerSW81b}, as well as recent lower bounds in fine-grained complexity~\cite{OkrasaR21,PiecykR21}.

When the $\HCol$ problems are parameterized by the number of vertices, they admit a trivial kernel of quadratic size. Therefore, the lower bounds implied by Theorem~\ref{thm:IntroLowerGen} for the vertex cover parameterization are limited to the near-quadratic regime. To surpass this barrier, we establish a lower bound on the kernel size of the $\HCol$ problem parameterized by the vertex cover number for all non-bipartite projective core graphs $H$ with $q(H) \geq 4$. We show that, for each such graph $H$, the degree of any polynomial that bounds the kernel size is unlikely to be smaller than the non-adjacency witness number of $H$ minus one.

\begin{theorem}\label{thm:IntroLowerGenQ(H)}
For every non-bipartite projective core graph $H$ with $q(H) \geq 4$ and for any real number $\eps >0$, the $\HCol$ problem parameterized by the vertex cover number $k$ does not admit a compression of size $O(k^{q(H)-1-\eps})$ unless $\NP \subseteq \coNPpoly$.
\end{theorem}

We note that the lower bound stated in Theorem~\ref{thm:IntroLowerGenQ(H)} is reasonably close to our upper bounds. Indeed, for every non-bipartite projective core graph $H$ with $q(H) \geq 4$, the upper bound given in Theorem~\ref{thm:Intro_kernel_q} exceeds the lower bound from Theorem~\ref{thm:IntroLowerGenQ(H)} by only a near-linear multiplicative factor. Since almost all graphs are non-bipartite projective cores, our results determine the kernel complexity of the problem up to a near-linear factor for almost every graph $H$. Furthermore, if a non-bipartite projective core graph $H$ admits a faithful independent representation of dimension $q(H) \geq 4$ over a computationally efficient field, then the kernel given in Theorem~\ref{thm:Intro_kernel_d} essentially meets the lower bound (see also Theorem~\ref{thm:kernel_IR}). As an illustrative consequence, we derive that the kernel sizes we achieve for all non-bipartite Kneser graphs are nearly tight (extending Item~\ref{itm:IntroLowerKneser} of Theorem~\ref{thm:IntroLower}; see Theorem~\ref{thm:LowerKneser}). The proof of Theorem~\ref{thm:IntroLowerGenQ(H)} relies on lower bounds on the compression size of certain satisfiability problems, extends the approach developed in~\cite{JansenK13} for coloring problems, and incorporates the aforementioned gadget construction from~\cite{OkrasaR21}.

\subsection{Outline}
The rest of the paper is organized as follows. In Section~\ref{sec:pre}, we collect several definitions and facts that will be used throughout the paper. In Section~\ref{sec:kernels}, we present and analyze our combinatorial and algebraic kernels for the $\HCol$ problem parameterized by the vertex cover number, thereby proving Theorems~\ref{thm:Intro_kernel_q} and~\ref{thm:Intro_kernel_d}. In Section~\ref{sec:lower}, we obtain lower bounds on the compression size of $\HCol$ problems with respect to two parameterizations, the number of vertices and the vertex cover number, confirming Theorems~\ref{thm:IntroLowerGen} and~\ref{thm:IntroLowerGenQ(H)}. In Section~\ref{sec:NAWN}, we conduct a systematic study of the non-adjacency witness number of graphs. In Section~\ref{sec:app}, we apply the results from the previous sections and present concrete upper and lower bounds on the kernel complexity of $\HCol$ problems parameterized by the vertex cover number for various graphs $H$. In particular, we confirm Theorems~\ref{thm:IntroOD} and~\ref{thm:IntroLower} therein. Concluding remarks and open questions appear in the final Section~\ref{sec:conclude}.

\section{Preliminaries}\label{sec:pre}

Throughout the paper, we omit floor and ceiling signs when they are not crucial, and all logarithms are in base $2$. For a positive integer $n$, we write $[n]$ to denote the set $\{1,2,\ldots,n\}$.

\subsection{Graphs}\label{sec:graphs}
The graphs considered in this paper are, unless stated otherwise, finite and simple, meaning they have no loops or parallel edges. When loops are allowed, this is stated explicitly, and we treat them as edges of the form $\{v\}$ for a vertex $v$. For a graph $G=(V,E)$ and a set $X \subseteq V$, we let $G[X]$ denote the subgraph of $G$ induced by $X$. The set $X$ is called a {\em clique} when the graph $G[X]$ is complete, and a {\em vertex cover} when the graph $G[V \setminus X]$ is edgeless. For a vertex $v \in V$, we let $N_G(v)$ denote the set of neighbors of $v$ in $G$. As is customary, we denote by $\Delta(G)$ the maximum degree of a vertex in $G$, and by $\omega(G)$ the maximum size of a clique in $G$. For a positive integer $d$, a graph $G$ is called {\em $d$-degenerate} if every subgraph of $G$ has a vertex of degree at most $d$. A graph is {\em planar} if it can be drawn in the plane without any edges crossing, except at their endpoints. It is well known that every planar graph is $5$-degenerate. The {\em chromatic number} of a graph is the minimum number of colors needed for a vertex coloring in which no two adjacent vertices share a color.

For two graphs $G = (V_G,E_G)$ and $H = (V_H,E_H)$, possibly with loops, a {\em homomorphism} from $G$ to $H$ is a mapping $f:V_G \rightarrow V_H$ such that for all vertices $u,v \in V_G$ with $\{u,v\} \in E_G$, it holds that $\{f(u),f(v)\} \in E_H$. If there exists a homomorphism from $G$ to $H$, the graph $G$ is said to be {\em $H$-colorable}. When $G$ and $H$ are simple, an {\em isomorphism} from $G$ to $H$ is a bijection $f:V_G \rightarrow V_H$ such that for all vertices $u,v \in V_G$, it holds that $\{u,v\} \in E_G$ if and only if $\{f(u),f(v)\} \in E_H$. If such an isomorphism exists, $G$ and $H$ are said to be {\em isomorphic}. An isomorphism from a graph to itself is called an {\em automorphism}. A {\em core} of a graph $G$ is a graph $H$ with a minimum number of vertices, such that there exist homomorphisms from $G$ to $H$ and from $H$ to $G$. It is known that a core of a graph $G$ is isomorphic to an induced subgraph of $G$. A graph is called a {\em core} if it is a core of itself. It is well known that a graph is a core if and only if every homomorphism from the graph to itself is an automorphism. For a detailed exposition of cores, see, e.g.,~\cite{HellN92}.

We introduce here a few important families of graphs. For a positive integer $m$, we denote by $K_m$ and $C_m$ the complete and cycle graphs on $m$ vertices, respectively. For positive integers $m$ and $r$ with $m \geq 2r$, the {\em Kneser graph} $K(m,r)$ is defined as the graph whose vertex set consists of all $r$-subsets of $[m]$, where two vertices are adjacent if and only if the corresponding subsets are disjoint.

\subsection{Linear Algebra}

For a field $\Fset$ and an integer $d$, two vectors $x,y \in \Fset^d$ are said to be {\em orthogonal} if $\langle x,y \rangle = 0$, with respect to the standard inner product defined by $\langle x,y \rangle = \sum_{i=1}^{d}{x_i y_i}$, where all operations are performed over $\Fset$. A vector $x \in \Fset^d$ is {\em self-orthogonal} if $\langle x,x \rangle = 0$, and {\em non-self-orthogonal} otherwise. We denote by $x^t$ the transpose of the column vector $x$, and use the same notation for the transpose of a matrix.

The {\em order} of a finite field is the number of its elements. It is well known that the order of any finite field is a prime power, and that every prime power is the order of some finite field. For two fields $\Fset$ and $\Kset$, we say that $\Kset$ is an {\em extension field} of $\Fset$ if $\Fset \subseteq \Kset$ and the operations of $\Fset$ agree with those of $\Kset$ when restricted to elements in $\Fset$. A finite field of order $p^m$, for a prime $p$ and a positive integer $m$, has an extension field of order $p^{m'}$ for every integer $m' \geq m$.

A multivariate polynomial over a field $\Fset$ is called {\em homogeneous of degree $d$} if each of its monomials has degree $d$, and is called {\em multilinear} if each of them forms a product of distinct variables. The set of all multilinear homogeneous polynomials of degree $d$ in $n$ variables over a field $\Fset$ forms a vector space over $\Fset$ of dimension $\binom{n}{d}$.

\subsection{Parameterized Complexity}

We collect here basic definitions on kernelization from the field of parameterized complexity. For an in-depth introduction to the topic, the reader is referred to, e.g.,~\cite{KernelBook19}.

A {\em parameterized problem} is a set $Q \subseteq \Sigma^* \times \N$ for some finite alphabet $\Sigma$. A {\em compression} (also known as generalized kernel) for a parameterized problem $Q \subseteq \Sigma^* \times \N$ into a parameterized problem $Q' \subseteq \Sigma^* \times \N$ is an algorithm that given an instance $(x,k) \in \Sigma^* \times \N$ returns in time polynomial in $|x|+k$ an instance $(x',k') \in \Sigma^* \times \N$, such that $(x,k) \in Q$ if and only if $(x',k') \in Q'$, and in addition, $|x'|+k' \leq h(k)$ for some computable function $h$. The function $h$ is called the {\em size} of the compression. If $|\Sigma|=2$, the function $h$ is called the {\em bit-size} of the compression. When we say that a parameterized problem $Q$ admits a compression of size $h$, we mean that there exists a compression of size $h$ for $Q$ into {\em some} parameterized problem. A compression for a parameterized problem $Q$ into itself is called a {\em kernelization}, or simply a {\em kernel}, for $Q$.

A {\em transformation} from a parameterized problem $Q \subseteq \Sigma^* \times \N$ into a parameterized problem $Q' \subseteq \Sigma^* \times \N$ is an algorithm that given an instance $(x,k) \in \Sigma^* \times \N$ returns in time polynomial in $|x|+k$ an instance $(x',k') \in \Sigma^* \times \N$, such that $(x,k) \in Q$ if and only if $(x',k') \in Q'$, and in addition, $k' \leq h(k)$ for some computable function $h$. If $h$ is linear, the transformation is called {\em linear-parameter}.

For a fixed graph $H$, the computational {\em $\HCol$} problem asks to decide whether an input graph $G$ is $H$-colorable. When parameterized by the vertex cover number, the input further includes a vertex cover of $G$, whose size is the parameter of the problem. While this definition is standard and convenient, we note that including the vertex cover $X$ in the input is not essential, as $X$ can be replaced by a vertex cover of $G$ computed via an efficient $2$-approximation algorithm.

\section{Kernels for \texorpdfstring{$\HCol$}{H-Coloring}}\label{sec:kernels}

In this section, we present and analyze our kernelization algorithms for the $\HCol$ problem parameterized by the vertex cover number. We begin with a combinatorial kernel, whose size is governed by the non-adjacency witness number of $H$, and then proceed to an algebraic kernel, whose size is bounded in terms of the dimension of a faithful independent representation of $H$.

\subsection{The Combinatorial Kernel}

The non-adjacency witness number of a graph measures the smallest number of vertices needed to certify that a given set of vertices has no common neighbor. We formally define it below and study it in detail in Section~\ref{sec:NAWN}.

\begin{definition}\label{def:q(G)}
The {\em non-adjacency witness number} of a graph $G = (V,E)$, denoted $q(G)$, is the smallest positive integer $q$ such that for every set $T \subseteq V$ of vertices with no common neighbor in $G$, there exists a set $T' \subseteq T$ of size $|T'| \leq q$ with no common neighbor in $G$.
\end{definition}

A main ingredient of our combinatorial kernel is the following lemma.

\begin{lemma}\label{lemma:G'}
Let $H=(V_H,E_H)$ be a graph, and let $q$ be an integer satisfying $q \geq q(H)$.
Consider a graph $G=(V,E)$ and a vertex cover $X \subseteq V$ of $G$ of size $k$.
Define $G'=(V',E')$ as the graph obtained from $G[X]$ by adding, for every non-empty set $S \subseteq X$ of size at most $q$ such that $S \subseteq N_G(v)$ for some $v \in V \setminus X$, a vertex $v_S$ adjacent to the vertices of $S$. Then the following holds.
\begin{enumerate}
  \item\label{itm:G'1} The set $X$ forms a vertex cover of $G'$.
  \item\label{itm:G'2} The number of vertices in $G'$ is at most $k + \sum_{i=1}^{q}{\binom{k}{i}}$.
  \item\label{itm:G'3} The graph $G'$ can be encoded in $\binom{k}{2}+ \sum_{i=1}^{q}{\binom{k}{i}}$ bits.
  \item\label{itm:G'4} The graph $G$ is $H$-colorable if and only if the graph $G'$ is $H$-colorable.
\end{enumerate}
\end{lemma}

\begin{proof}
By the definition of the graph $G'$, every edge of $G'$ connects either two vertices of $X$ or a vertex of $X$ to a vertex $v_S$, hence $X$ forms a vertex cover of $G'$.
It also follows from the definition that the number of vertices in $G'$ does not exceed $k + \sum_{i=1}^{q}{\binom{k}{i}}$. Additionally, $G'$ can be represented by a binary string that expresses the adjacencies inside $G[X]$ as well as whether the vertex $v_S$ is included in $G'$ or not, for each non-empty set $S \subseteq X$ of size $|S| \leq q$. Therefore, the graph $G'$ can be encoded in $\binom{k}{2}+ \sum_{i=1}^{q}{\binom{k}{i}}$ bits.

It remains to prove the fourth item of the lemma, namely, that the graph $G$ is $H$-colorable if and only if the graph $G'$ is $H$-colorable.
Suppose first that $G$ is $H$-colorable, and let $f:V \rightarrow V_H$ be a homomorphism from $G$ to $H$.
Consider the function $f':V' \rightarrow V_H$ defined as follows. First, for every vertex $u \in X$, let $f'(u) = f(u)$. Since every edge of $G'[X]$ is also an edge of $G[X]$, its endpoints are mapped by $f'$ to adjacent vertices in $H$. Now, every vertex in $V' \setminus X$ has the form $v_S$ for a set $S \subseteq X$ of size $|S| \leq q$, where there exists a vertex $v \in V \setminus X$ with $S \subseteq N_G(v)$. We define $f'(v_S) = f(v)$ for such a vertex $v$. Since $f$ and $f'$ agree on $X$, it follows that $f'(v_S)$ is adjacent in $H$ to $f'(u)$ for all vertices $u \in S = N_{G'}(v_S)$. We thus obtain that $f'$ is a homomorphism from $G'$ to $H$, so $G'$ is $H$-colorable.

For the converse direction, suppose that $G'$ is $H$-colorable, and let $f':V' \rightarrow V_H$ be a homomorphism from $G'$ to $H$. Consider the function $f:V \rightarrow V_H$ defined as follows. First, for every vertex $u \in X$, let $f(u) = f'(u)$. Since every edge of $G[X]$ is also an edge of $G'[X]$, its endpoints are mapped by $f$ to adjacent vertices in $H$. Next, consider any vertex $u \in V \setminus X$. Since $X$ is a vertex cover of $G$, all the neighbors of $u$ lie in $X$. We claim that there exists a vertex in $H$ that is adjacent to all vertices $f(v)$ with $v \in N_G(u)$. To see this, let $T = \{f(v) \mid v \in N_G(u)\}$, and suppose for contradiction that the vertices of $T$ have no common neighbor in $H$. By the definition of the non-adjacency witness number, there exists a non-empty set $T' \subseteq T$ of size $|T'| \leq q(H) \leq q$ with no common neighbor in $H$. Notice that there also exists a non-empty set $S \subseteq N_G(u)$ of size $|S| \leq q$ such that $T' = \{f(v) \mid v \in S\}$. However, the graph $G'$ includes the corresponding vertex $v_S$, which is mapped by $f'$ to a vertex in $H$ that is adjacent to all vertices of $T'$, a contradiction. This shows that the set $T$ must admit a common neighbor in $H$, allowing us to define $f(u)$ as such a neighbor. It follows that $f$ is a homomorphism from $G$ to $H$, so $G$ is $H$-colorable, as desired.
\end{proof}

We are ready to derive Theorem~\ref{thm:Intro_kernel_q}, which states that for every graph $H$, the $\HCol$ problem parameterized by the vertex cover number $k$ admits a kernel with $O(k^{q(H)})$ vertices and bit-size $O(k^{q(H)})$. Consequences for concrete graph classes appear in Section~\ref{sec:app}.

\begin{proof}[ of Theorem~\ref{thm:Intro_kernel_q}]
Let $H=(V_H,E_H)$ be a fixed graph, and set $q = q(H)$. If $H$ is edgeless, then the $\HCol$ problem is equivalent to determining whether an input graph is edgeless, and is thus solvable in polynomial time. Therefore, we may and will assume that $H$ has at least one edge, which easily implies that $q \geq 2$ (see Lemma~\ref{lemma:q_bounds}).

The input of the $\HCol$ problem parameterized by the vertex cover number $k$ consists of a graph $G=(V,E)$ and a vertex cover $X \subseteq V$ of $G$ of size $|X|=k$.
Consider the kernelization algorithm that, given such an input, produces the graph $G'=(V',E')$ defined as in Lemma~\ref{lemma:G'}, that is, the graph obtained from $G[X]$ by adding, for every non-empty set $S \subseteq X$ of size $|S| \leq q$ such that $S \subseteq N_G(v)$ for some $v \in V \setminus X$, a vertex $v_S$ adjacent to the vertices of $S$. Then, the algorithm returns the graph $G'$ along with the set $X$.

The analysis of the algorithm is based on Lemma~\ref{lemma:G'}. By Item~\ref{itm:G'1} of the lemma, the set $X$ forms a vertex cover of $G'$, so the pair $(G',X)$ produced by the algorithm is a valid output. By Item~\ref{itm:G'2}, the number of vertices in $G'$ satisfies $|V'| \leq k + \sum_{i=1}^{q}{\binom{k}{i}} \leq O(k^q)$, and by Item~\ref{itm:G'3}, the number of bits needed to encode $G'$ is at most $\binom{k}{2}+ \sum_{i=1}^{q}{\binom{k}{i}}$, which is $O(k^q)$ since $q \geq 2$. Item~\ref{itm:G'4} ensures the correctness of the kernelization algorithm. Finally, since $q$ is a fixed constant, one can enumerate all subsets of $X$ of size at most $q$ and construct the graph $G'$ in polynomial time. The proof is now complete.
\end{proof}

\subsection{The Algebraic Kernel}

Our algebraic kernel for the $\HCol$ problem parameterized by the vertex cover number uses the notion of independent representations, a generalization of orthogonal representations of graphs.

\subsubsection{Orthogonal and Independent Graph Representations}

We begin with formal definitions of these two representation types.

\begin{definition}\label{def:OR}
A {\em $d$-dimensional orthogonal representation} of a graph $G=(V,E)$ over a field $\Fset$ is an assignment of a vector $x_v \in \Fset^d$ with $\langle x_v,x_v\rangle \neq 0$ to each vertex $v \in V$, such that for all vertices $u,v \in V$, if $\{u,v\} \in E$ then $\langle x_u,x_v \rangle =0$. The orthogonal representation is called {\em faithful} if for all vertices $u,v \in V$, it holds that $\{u,v\} \in E$ if and only if $\langle x_u,x_v \rangle =0$.
\end{definition}

\begin{definition}\label{def:IR}
A {\em $d$-dimensional independent representation} of a graph $G=(V,E)$ over a field $\Fset$ is an assignment of a vector $x_v \in \Fset^d$ to each vertex $v \in V$, such that for every vertex $v \in V$, it holds that $x_v \notin \linspan(\{x_w \mid w \in N_G(v)\})$. The independent representation is called {\em faithful} if for all vertices $u,v \in V$, it holds that $\{u,v\} \in E$ if and only if $x_u \in \linspan(\{x_w \mid w \in N_G(v)\})$. Notice that the forward implication holds trivially.
\end{definition}
\noindent
We note that all vectors in an independent representation of a graph are nonzero. It can be verified that a graph has a $2$-dimensional independent representation over a field if and only if it is bipartite.

\begin{remark}\label{remark:OD_ID}
Every (faithful) orthogonal representation of a graph $G=(V,E)$ over a field $\Fset$ is also a (faithful) independent representation of $G$ over $\Fset$. To see this, consider an orthogonal representation $(x_v)_{v \in V}$ of $G$ over $\Fset$. For each vertex $v \in V$, the vector $x_v$ is orthogonal to the vectors $x_w$ with $w \in N_G(v)$, and thus to the subspace they span, but not to itself, so $x_v \notin \linspan(\{x_w \mid w \in N_G(v)\})$. This shows that $(x_v)_{v \in V}$ forms an independent representation of $G$.
If the orthogonal representation $(x_v)_{v \in V}$ of $G$ is faithful, then for all vertices $u,v \in V$ with $\{u,v\} \notin E$, the vector $x_v$ is not orthogonal to $x_u$ but is orthogonal to the vectors $x_w$ with $w \in N_G(v)$, implying that $x_u \notin \linspan(\{x_w \mid w \in N_G(v)\})$. This confirms that $(x_v)_{v \in V}$ forms a faithful independent representation of $G$ in this case.
\end{remark}

We proceed by presenting several lemmas that will be used in the analysis of the algebraic kernel. The first lemma shows that the dimension of any faithful independent representation of a graph forms an upper bound on its non-adjacency witness number (recall Definition~\ref{def:q(G)}).

\begin{lemma}\label{lemma:q_faith}
For a graph $G$, a positive integer $d$, and a field $\Fset$, if $G$ has a faithful $d$-dimensional independent representation over $\Fset$, then $q(G)\leq d$.
\end{lemma}

\begin{proof}
Let $G=(V,E)$ be a graph, and let $(x_v)_{v \in V}$ be a faithful $d$-dimensional independent representation of $G$ over a field $\Fset$.
To prove that $q(G) \leq d$, consider a set $T \subseteq V$ of vertices with no common neighbor in $G$, and let $W = \linspan(\{x_u \mid u \in T\})$ be the subspace spanned by their associated vectors. Define $T' \subseteq T$ as a subset of vertices of $T$ whose associated vectors form a basis of $W$, and note that $W = \linspan(\{x_u \mid u \in T'\})$. Since the subspace $W$ lies in $\Fset^d$, we clearly have $|T'| \leq d$. We claim that the vertices of $T'$ have no common neighbor in $G$. Indeed, suppose for contradiction that there exists a vertex $v \in V$ adjacent to all vertices of $T'$. Then, for every vertex $u \in T'$, it holds that $x_u \in \linspan(\{x_w \mid w \in N_G(v)\})$, and thus $W \subseteq \linspan(\{x_w \mid w \in N_G(v)\})$. This implies that for every vertex $u \in T$, it holds that $x_u \in W \subseteq \linspan(\{x_w \mid w \in N_G(v)\})$, so by the faithfulness of the given representation, $u$ must be adjacent to $v$ in $G$. Therefore, $v$ is a common neighbor of the vertices of $T$, contradicting our assumption. To conclude, every set of vertices with no common neighbor in $G$ has a subset of size at most $d$ with no common neighbor, so $q(G) \leq d$, as desired.
\end{proof}

We next show that faithful independent representations are invariant under invertible linear transformations over extension fields.

\begin{lemma}\label{lemma:Ax_v}
Let $G= (V,E)$ be a graph, and let $\Fset \subseteq \Kset$ be fields such that $\Kset$ is an extension field of $\Fset$. For a positive integer $d$, let $(x_v)_{v \in V}$ be a faithful $d$-dimensional independent representation of $G$ over $\Fset$, and let $A \in \Kset^{d \times d}$ be an invertible matrix. Then $(A x_v)_{v \in V}$ is a faithful $d$-dimensional independent representation of $G$ over $\Kset$.
\end{lemma}

\begin{proof}
Let $A \in \Kset^{d \times d}$ be an invertible matrix.
We first observe that if $\ell$ vectors $x_1, \ldots, x_\ell \in \Fset^d$ are linearly independent over $\Fset$, then the vectors $Ax_1, \ldots, Ax_\ell \in \Kset^d$ are linearly independent over $\Kset$. To see this, complete $x_1, \ldots, x_\ell$ to a basis of $\Fset^d$, and let $B \in \Fset^{d \times d}$ denote the matrix whose columns form this basis. Since the determinant of $B$ over $\Fset$ is nonzero, it remains nonzero when interpreted over $\Kset$. This implies that the first $\ell$ columns of $B$, namely $x_1, \ldots, x_\ell$, are linearly independent over $\Kset$. Now, suppose that $\sum_{i=1}^{\ell}{\alpha_i Ax_i} = 0$ for some coefficients $\alpha_1,\ldots,\alpha_\ell \in \Kset$. It follows that $A \cdot \sum_{i=1}^{\ell}{\alpha_i x_i} = 0$, and by the invertibility of $A$, we get that $\sum_{i=1}^{\ell}{\alpha_i x_i} = 0$. Since $x_1, \ldots, x_\ell$ are linearly independent over $\Kset$, we conclude that $\alpha_i=0$ for all $i \in [\ell]$, as required.

Now, let $(x_v)_{v \in V}$ be a faithful $d$-dimensional independent representation of $G$ over $\Fset$.
To prove that $(A x_v)_{v \in V}$ is a faithful independent representation of $G$ over $\Kset$, consider two non-adjacent vertices $u$ and $v$ in $G$. Define $W_v = \linspan(\{x_w \mid w \in N_G(v)\}) \subseteq \Fset^d$, and let $S \subseteq N_G(v)$ be a set of neighbors of $v$ such that the vectors $x_w$ with $w \in S$ form a basis of $W_v$. Since the given independent representation is faithful, it follows that $x_u \notin W_v$, so the set $\{x_u\} \cup \{x_w \mid w \in S\}$ is linearly independent over $\Fset$. By the observation above, this implies that the set $\{Ax_u\} \cup \{Ax_w \mid w \in S\}$ is linearly independent over $\Kset$. Hence, $Ax_u \notin \linspan(\{Ax_w \mid w \in S\}) = \linspan(\{Ax_w \mid w \in N_G(v)\})$, and we are done.
\end{proof}

We now show that any faithful independent representation over a field can be transformed, over any sufficiently large extension field, into one of the same dimension in which all vectors start with $1$.

\begin{lemma}\label{lemma:IR_with1}
Let $G= (V,E)$ be a graph, let $d$ be a positive integer, and let $\Fset \subseteq \Kset$ be fields such that $\Kset$ is an extension field of $\Fset$ with $|\Kset| > |V|$. If $G$ has a faithful $d$-dimensional independent representation over $\Fset$, then it also has a faithful $d$-dimensional independent representation over $\Kset$, in which all vectors have $1$ as their first entry.
\end{lemma}

\begin{proof}
Let $(x_v)_{v \in V}$ be a faithful $d$-dimensional independent representation of $G$ over $\Fset$. We first show that there exists a vector $y \in \Kset^d$ such that $\langle y , x_v \rangle \neq 0$ for all $v \in V$. This is proved using the probabilistic method. By $|\Kset| > |V|$, one can choose a set $K \subseteq \Kset$ such that $|K| = |V|+1$. Let $y \in \Kset^d$ be a random vector whose entries are chosen uniformly and independently from $K$. We observe that for each vertex $v \in V$, the probability that $\langle y, x_v \rangle = 0$ is at most $\frac{1}{|K|}$. To see this, let $j \in [d]$ be an index for which the $j$th entry of $x_v$ is nonzero. Then, for any fixed values of $y_i$ with $i \in [d] \setminus \{j\}$, there is at most one choice for $y_j$ such that $\langle y,x_v \rangle = 0$. By the union bound, the probability that there exists a vertex $v \in V$ such that $\langle y,x_v \rangle = 0$ is at most $\frac{|V|}{|K|} = \frac{|V|}{|V|+1} <1$, so with positive probability, the inequality $\langle y , x_v \rangle \neq 0$ holds for all $v \in V$. This implies the existence of the desired vector $y$.

Now, fix a vector $y \in \Kset^d$ as above, satisfying $\langle y , x_v \rangle \neq 0$ for all $v \in V$. Since $y$ is nonzero, there exists an invertible matrix $A \in \Kset^{d \times d}$ whose first row is $y^t$. By Lemma~\ref{lemma:Ax_v}, the assignment $(A x_v)_{v \in V}$ forms a faithful $d$-dimensional independent representation of $G$ over $\Kset$. Note that for each vertex $v \in V$, the first entry of $Ax_v$ is $\langle y,x_v \rangle$, and is thus nonzero. We may therefore scale each such vector to make its first entry equal to $1$, obtaining the desired faithful independent representation of $G$ over $\Kset$.
\end{proof}

We finally observe that every graph admits a faithful independent representation of dimension one greater than its maximum degree over any sufficiently large field.

\begin{lemma}\label{lemma:faith_Delta+1}
Let $G= (V,E)$ be a graph, let $\Fset$ be a field with $|\Fset| \geq |V|$, and set $d = \Delta(G)+1$.
Then $G$ has a faithful $d$-dimensional independent representation over $\Fset$.
\end{lemma}

\begin{proof}
For a graph $G=(V,E)$, set $d = \Delta(G)+1$ and $n = |V|$, and let $v_1, \ldots, v_n$ denote the vertices of $G$.
Let $\Fset$ be a field with $|\Fset| \geq n$, and let $\alpha_1, \ldots,\alpha_n \in \Fset$ denote $n$ distinct field elements.
For each $i \in [n]$, assign to the vertex $v_i$ the vector $x_{v_i} = (1,\alpha_i,\alpha_i^2, \ldots, \alpha_i^{d-1})^t \in \Fset^d$. Observe that any $d$ vectors in the assignment $(x_{v})_{v \in V}$ form the columns of an invertible Vandermonde matrix, and are therefore linearly independent over $\Fset$.
We claim that the assignment $(x_{v})_{v \in V}$ forms a faithful $d$-dimensional independent representation of $G$ over $\Fset$.
To verify this, consider indices $i,j \in [n]$ for which the vertices $v_i$ and $v_j$ are not adjacent in $G$, and let $W = \linspan(\{ x_{v} \mid v \in N_G(v_j)\})$. The degree of $v_j$ does not exceed $\Delta(G)$, so the subspace $W$ is spanned by at most $d-1$ vectors $x_v$ with $v \in V \setminus \{v_i\}$. Since any $d$ vectors $x_v$ with $v \in V$ are linearly independent, it follows that $x_{v_i} \notin W$, which completes the argument.
\end{proof}

\begin{remark}\label{remark:IR_def}
We conclude this section with the observation that the existence of a faithful $d$-dimensional independent representation of a graph $G=(V,E)$ over a sufficiently large field $\Fset$ is equivalent to the existence of a matrix $M \in \Fset^{|V| \times |V|}$ of rank at most $d$, with rows and columns indexed by the vertex set $V$, such that for all $u,v \in V$, it holds that $\{u,v\} \in E$ if and only if $M_{u,v} = 0$ (in particular, all diagonal entries of $M$ are nonzero). To see this, consider such a matrix $M$, and write $M = X^t \cdot Y$ for two matrices $X,Y \in \Fset^{d \times |V|}$. For each vertex $v \in V$, let $x_v$ and $y_v$ denote, respectively, the columns of $X$ and $Y$ associated with $v$. If $u$ and $v$ are non-adjacent vertices in $G$, then $\langle x_u,y_v \rangle \neq 0$, whereas for every vertex $w \in N_G(v)$, it holds that $\langle x_w,y_v \rangle =0$, implying that $x_u \notin \linspan(\{x_w \mid w\in N_G(v)\})$. Therefore, the assignment $(x_v)_{v \in V}$ forms a faithful $d$-dimensional independent representation of $G$ over $\Fset$. Conversely, suppose that $G$ has a faithful $d$-dimensional independent representation $(x_v)_{v \in V}$ over $\Fset$. Assuming that $\Fset$ is sufficiently large, say $|\Fset| > |V|$, a simple probabilistic argument shows that for every vertex $v \in V$, there exists a vector $y_v \in \Fset^d$ such that for every vertex $u \in V$, $\langle x_u,y_v \rangle = 0$ if and only if $\{u,v\} \in E$. Let $X,Y \in \Fset^{d \times |V|}$ denote the matrices whose columns corresponding to each vertex $v$ are $x_v$ and $y_v$, respectively. Then the matrix $M=X^t \cdot Y$ has rank at most $d$, and it holds that $\{u,v\} \in E$ if and only if $M_{u,v} = 0$ for all vertices $u,v \in V$, as required.
\end{remark}

\subsubsection{The Algorithm}

We are ready to present our algebraic kernel and confirm Theorem~\ref{thm:Intro_kernel_d}. To state the result in full generality, we borrow the following terminology from~\cite{JansenP19sparse}. A field is said to be {\em efficient} if field operations and Gaussian elimination can be performed in polynomial time in the size of a reasonable input encoding. Note that all finite fields and the real field $\R$ (restricted to rational numbers) are efficient.

\begin{theorem}\label{thm:kernel_IR}
For a graph $H$, a positive integer $d$, and an efficient field $\Fset$, suppose that there exists a faithful $d$-dimensional independent representation of $H$ over $\Fset$. Then the $\HCol$ problem parameterized by the vertex cover number $k$ admits a kernel with $O(k^{d-1})$ vertices and bit-size $O(k^{d-1} \cdot \log k)$.
\end{theorem}

\begin{proof}
Consider a graph $H=(V_H,E_H)$, a positive integer $d$, and an efficient field $\Fset$, such that $H$ has a faithful $d$-dimensional independent representation over $\Fset$. We may assume that $d \geq 3$, as otherwise $H$ is bipartite, and the $\HCol$ problem can be solved in polynomial time. We define a field $\Kset$ as follows. If $\Fset$ is finite, then we choose $\Kset$ to be some finite extension field of $\Fset$ whose order exceeds $|V_H|$ (e.g., a finite field of order $|\Fset|^\ell$ for the smallest positive integer $\ell$ satisfying $|\Fset|^\ell > |V_H|$). If $\Fset$ is infinite, then $\Kset$ is simply taken to be $\Fset$. In both cases, the field $\Kset$ is efficient and satisfies $|\Kset| > |V_H|$.

The input of the $\HCol$ problem parameterized by the vertex cover number $k$ consists of a graph $G=(V,E)$ and a vertex cover $X \subseteq V$ of $G$ of size $|X|=k$. Consider the kernelization algorithm that, given such an input, performs the following steps.
\begin{enumerate}
  \item Produce the graph $G'=(V',E')$ defined as in Lemma~\ref{lemma:G'} for $q=d$. Namely, $G'$ is the graph obtained from $G[X]$ by adding, for every non-empty set $S \subseteq X$ of size $|S| \leq d$ such that $S \subseteq N_G(v)$ for some $v \in V \setminus X$, a vertex $v_S$ adjacent to the vertices of $S$.
  \item Associate with each vertex $u \in X$ a $d$-dimensional symbolic vector $y_u$ over the field $\Kset$, where the first entry of $y_u$ is fixed to $1$ and the remaining $d-1$ entries are free variables. Note that the total number of variables is $k \cdot (d-1)$.
  \item\label{itm:p_S} Consider the collection $Y$ of all sets $S \subseteq X$ of size precisely $d$ for which the vertex $v_S$ is included in the graph $G'$, i.e., $Y = \{ S \subseteq X \mid v_S \in V',~|S|=d \}$. For each such set $S \in Y$, define a polynomial $p_S$ over the variables of $(y_u)_{u \in X}$ as the determinant over the field $\Kset$ of a $d \times d$ matrix whose columns are the vectors $y_u$ with $u \in S$ (ordered arbitrarily).
  \item Consider the subspace of polynomials $W= \linspan(\{p_S \mid S \in Y\})$, and compute a set $Y' \subseteq Y$ such that the polynomials in $\{p_S \mid S \in Y'\}$ form a basis of $W$.
  \item Let $G''=(V'',E'')$ be the graph obtained from $G'$ by removing all vertices $v_S$ with $S \in Y \setminus Y'$, and return it along with the set $X$.
\end{enumerate}
Let us emphasize that the given independent representation of $H$ is not used by the algorithm itself, but it plays a crucial role in its correctness proof. Yet, the algorithm uses the field $\Kset$, over which it computes the polynomials $p_S$ for $S \in Y$ and a basis of the subspace $W$.

We now analyze the algorithm and show that it satisfies the assertion of the theorem.
First, by Item~\ref{itm:G'1} of Lemma~\ref{lemma:G'}, the set $X$ forms a vertex cover of $G'$, and so also forms a vertex cover of its subgraph $G''$. Therefore, the pair $(G'',X)$ produced by the algorithm is a valid output. The vertex set of $G''$ consists of the $k$ vertices of $X$, vertices $v_S$ with $|S| \leq d-1$, whose number is at most $\sum_{i=1}^{d-1}{\binom{k}{i}}$, and the vertices $v_S$ with $S \in Y'$. To bound the size of $Y'$, we observe that each polynomial $p_S$ with $S \in Y$ is multilinear and homogeneous of degree $d-1$. Indeed, the determinant polynomial of a $d \times d$ matrix is a linear combination of monomials of degree $d$, where each monomial is a product of exactly $d$ entries, one from each row. Since $p_S$ is defined as the determinant of a $d \times d$ matrix, whose first-row entries are all ones and all other entries are distinct variables, it follows that $p_S$ has the required form. This implies that $W$ is a subspace of the space of multilinear homogeneous polynomials of degree $d-1$ in $k \cdot (d-1)$ variables, whose dimension is $\binom{k \cdot (d-1)}{d-1}$. Since the size of a basis of such a subspace cannot exceed the dimension of the ambient space, it follows that $|Y'| \leq \binom{k \cdot (d-1)}{d-1}$. We obtain that the number of vertices in $G''$ satisfies
\[|V''| \leq k+ \sum_{i=1}^{d-1}{\binom{k}{i}} + \binom{k \cdot (d-1)}{d-1} \leq O(k^{d-1}),\]
where the last inequality holds because $d$ is a fixed constant. The graph $G''$ can be expressed by a binary string that represents the adjacencies in $G''[X]$ and the sets $S \subseteq X$ with $|S| \leq d$ for which the vertex $v_S$ is included in $G''$. Since each such set $S$ can be encoded in $d \cdot \lceil \log k \rceil$ bits, the total number of bits needed to encode $G''$ does not exceed $\binom{k}{2} + |V'' \setminus X| \cdot d \cdot \lceil \log k \rceil$, which is bounded by $O(k^{d-1} \cdot \log k)$, given that $d \geq 3$. We further notice that $G''$ can be constructed in polynomial time. Indeed, the construction of $G'$ involves enumerating all subsets of $X$ of size at most $d$. Then, to determine the vertices retained in $G''$, we compute the set $Y'$ by applying Gaussian elimination over the efficient field $\Kset$ on a system with $\binom{k \cdot (d-1)}{d-1}$ variables, which can be performed in polynomial time.

We turn to proving the correctness of the algorithm, namely, that $G$ is $H$-colorable if and only if $G''$ is $H$-colorable. By assumption, the graph $H$ has a faithful $d$-dimensional independent representation over $\Fset$. Since $\Kset$ is an extension field of $\Fset$ with $|\Kset| > |V_H|$, Lemma~\ref{lemma:IR_with1} implies that $H$ also has a faithful $d$-dimensional independent representation $(x_v)_{v \in V_H}$ over $\Kset$, in which all vectors have $1$ as their first entry. By Lemma~\ref{lemma:q_faith}, we have $q(H) \leq d$, which allows us to apply Item~\ref{itm:G'4} of Lemma~\ref{lemma:G'} and conclude that the graph $G$ is $H$-colorable if and only if the graph $G'$ is $H$-colorable. It is thus sufficient to prove that $G'$ is $H$-colorable if and only if $G''$ is $H$-colorable. Since $G''$ is a subgraph of $G'$, it is clear that if $G'$ is $H$-colorable then so is $G''$. It remains to show that if $G''$ is $H$-colorable, then so is $G'$.

Suppose that the graph $G''$ is $H$-colorable, and let $f'':V'' \rightarrow V_H$ be a homomorphism from $G''$ to $H$.
We define an assignment $\rho$ of elements from $\Kset$ to the variables of $(y_u)_{u \in X}$ by assigning to each variable vector $y_u$ the vector $x_{f''(u)}$ from the given independent representation over $\Kset$. Note that this is possible because the first entry in $y_u$ is $1$, just as in $x_{f''(u)}$, while the other entries in $y_u$ are free variables. We claim that for every set $S \in Y$, the polynomial $p_S$ defined in Item~\ref{itm:p_S} of the algorithm vanishes at the assignment $\rho$. It suffices to verify this only for sets $S \in Y'$, since the set $\{p_S \mid S \in Y'\}$ forms a basis of the subspace $W= \linspan(\{p_S \mid S \in Y\})$. Consider a set $S \in Y'$, and recall that the vertex $v_S$ is adjacent in $G''$ to the vertices of $S$, hence $f''(v_S)$ is adjacent in $H$ to all vertices $f''(u)$ with $u \in S$. By the definition of an independent representation, the vector $x_{f''(v_S)}$ does not lie in the linear subspace spanned by the vectors $x_{f''(u)}$ with $u \in S$. In particular, the dimension of this subspace is strictly smaller than $d$, so the determinant of a $d \times d$ matrix whose columns are the vectors $x_{f''(u)}$ with $u \in S$ is zero. This implies that $p_S$ evaluates to zero at the assignment $\rho$, as desired.

We now show that there exists a homomorphism $f':V' \rightarrow V_H$ from $G'$ to $H$.
We start by defining $f'(u) = f''(u)$ for each $u \in X$. Since every edge of $G'[X]$ is also an edge of $G''[X]$, its endpoints are mapped by $f'$ to adjacent vertices in $H$. Next, for every vertex $v_S$ in $G'$ with $|S| \leq d-1$, set $f'(v_S)=f''(v_S)$. Since the neighborhood of such a vertex $v_S$ is contained in $X$ and is identical to its neighborhood in $G''$, the vertex $f'(v_S)$ is adjacent in $H$ to all the images under $f'$ of the neighbors of $v_S$ in $G'$. It remains to show that for every set $S \in Y$, the vertex $v_S$ in $G'$ can be assigned a vertex in $H$ that is adjacent to all vertices $f'(u)$ with $u \in S$.

To see this, consider a set $S \in Y$, and recall that $|S|=d$ and that the polynomial $p_S$ vanishes at the assignment $\rho$ described above. Therefore, the vectors of $\{x_{f'(w)} \mid w \in S\}$ span a subspace of dimension smaller than $d$, so there exists a non-empty set $S' \subseteq S$ of size $|S'| \leq d-1$ such that the vectors of $\{x_{f'(w)} \mid w \in S'\}$ span the exact same subspace. Since the vertex $v_S$ is included in $G''$, it follows that $S \subseteq N_G(v)$ for some $v \in V \setminus X$, so by $S' \subseteq S$, the vertex $v_{S'}$ is included in $G''$ as well. We define $f'(v_S) = f''(v_{S'})$, and show that $f''(v_{S'})$ is adjacent in $H$ to all vertices $f'(u)$ with $u \in S$. Indeed, for every vertex $u \in S$, it holds that
\begin{eqnarray}\label{eq:f'(u)}
x_{f'(u)} \in \linspan(\{x_{f'(w)} \mid w \in S\}) = \linspan(\{x_{f'(w)} \mid w \in S'\}) \subseteq \linspan(\{x_{v} \mid v \in N_H(f''(v_{S'}))\}),
\end{eqnarray}
where the containment holds because the vertices of $S'$ are adjacent in $G''$ to $v_{S'}$, so they are mapped by $f''$, and thus by $f'$, to neighbors of $f''(v_{S'})$ in $H$. By the faithfulness of the given independent representation of $H$, it follows from~\eqref{eq:f'(u)} that the vertices $f'(u)$ and $f''(v_{S'})$ are adjacent in $H$. We conclude that $f''(v_{S'})$ is adjacent in $H$ to all vertices $f'(u)$ with $u \in S$, as desired. This gives us a homomorphism from $G'$ to $H$, implying that $G'$ is $H$-colorable, as required.
\end{proof}

We end this section with the observation that Theorem~\ref{thm:kernel_IR} strengthens the kernelization bounds implied by the prior work~\cite{JansenP19color}. Indeed, by Lemma~\ref{lemma:faith_Delta+1}, every graph $H$ has a faithful independent representation of dimension $\Delta(H)+1$ over, for instance, the real field $\R$. Invoking Theorem~\ref{thm:kernel_IR} yields a kernel for the $\HCol$ problem parameterized by the vertex cover number $k$ with $O(k^{\Delta(H)})$ vertices and bit-size $O(k^{\Delta(H)} \cdot \log k)$. As illustrated in Section~\ref{sec:app}, the kernel obtained via Theorem~\ref{thm:kernel_IR} is sometimes markedly more compact.

\section{Lower Bounds}\label{sec:lower}

This section establishes conditional lower bounds on the compressibility of $\HCol$ problems parameterized by the number of vertices and by the vertex cover number.
As previously mentioned, it suffices to consider target graphs $H$ that are non-bipartite cores. We focus here on the case where the graph $H$ is projective. While the definition of projective graphs can be found in~\cite{LaroseT00}, we rely solely on a key property proved by Okrasa and Rz{\k a}\.{z}ewski~\cite{OkrasaR21}, stated as follows.

\begin{proposition}[{\cite[Lemma~4.1]{OkrasaR21}}]\label{prop:projective}
For every projective core graph $H=(V_H,E_H)$ with $|V_H| \geq 3$, there exists an edge gadget, that is, a graph $F=(V_F,E_F)$ with two specified vertices $a,b \in V_F$, such that for all pairs of vertices $u,v \in V_H$, there exists a homomorphism $g:V_F \rightarrow V_H$ from $F$ to $H$ with $g(a)=u$ and $g(b)=v$ if and only if $u \neq v$.
\end{proposition}
\noindent
Interestingly, it was further shown in~\cite{OkrasaR21} that, for connected core graphs on at least three vertices, projectivity is a necessary condition for the existence of an edge gadget. To illustrate this notion, one may verify that for every positive integer $m$, the path graph on $2m$ vertices, with its endpoints serving as the specified vertices, forms an edge gadget for the cycle $C_{2m+1}$ (see, e.g.,~\cite{Nesetril81}).

\subsection{Parameterization by Number of Vertices}

Our starting point is the $\ListHCol$ problem, defined as follows.

\begin{definition}\label{def:list}
For a fixed graph $H=(V_H,E_H)$, which may include loops, the $\ListHCol$ problem is defined as follows.
The input consists of a graph $G=(V,E)$ and a list $L(v) \subseteq V_H$ for each vertex $v \in V$, and the goal is to decide whether there exists a homomorphism $f:V \rightarrow V_H$ from $G$ to $H$ such that $f(v) \in L(v)$ for all $v \in V$.
When parameterized by the vertex cover number, the input also includes a vertex cover of $G$, and its size serves as the problem's parameter.
\end{definition}

Feder et al.~\cite{FederHH99} showed that the $\ListHCol$ problem is solvable in polynomial time if $H$ is a bi-arc graph, and $\NP$-hard otherwise. We omit here the definition of bi-arc graphs, as it is not essential to our discussion. We note, however, that all loopless non-bipartite graphs $H$, namely, those for which the (non-list) $\HCol$ problem was shown to be $\NP$-hard in~\cite{HellN90}, are not bi-arc. A result by Chen et al.~\cite{ChenJOPR23}, stated below, shows that for all graphs $H$ that are not bi-arc, the $\ListHCol$ problem is unlikely to admit a sub-quadratic compression.

\begin{theorem}[{\cite[Theorem~1.1]{ChenJOPR23}}]\label{thm:ChenJOPR}
Let $H$ be a graph, possibly with loops, that is not bi-arc.
Then, for any real number $\eps >0$, the $\ListHCol$ problem parameterized by the number of vertices $n$ does not admit a compression of size $O(n^{2-\eps})$ unless $\NP \subseteq \coNPpoly$.
\end{theorem}

The authors of~\cite{ChenJOPR23} asked whether an analogous lower bound can be established for the $\NP$-hard cases of the (non-list) $\HCol$ problem. We prove now Theorem~\ref{thm:IntroLowerGen}, which answers this question affirmatively for all non-bipartite projective core graphs $H$.


\begin{proof}[ of Theorem~\ref{thm:IntroLowerGen}]
Fix a non-bipartite projective core graph $H=(V_H,E_H)$ and a real number $\eps >0$.
Since the (simple) graph $H$ is not bipartite, it is not a bi-arc graph, so by Theorem~\ref{thm:ChenJOPR}, the $\ListHCol$ problem parameterized by the number of vertices $n$ does not admit a compression of size $O(n^{2-\eps})$ unless $\NP \subseteq \coNPpoly$. In what follows, we present a linear-parameter transformation from the $\ListHCol$ problem into the $\HCol$ problem, both parameterized by the number of vertices. Therefore, if the $\HCol$ problem parameterized by the number of vertices $n$ admits a compression of size $O(n^{2-\eps})$, then composing this compression with the transformation described below would yield a compression of size $O(n^{2-\eps})$ for the $\ListHCol$ problem, implying that $\NP \subseteq \coNPpoly$.

Since $H$ is not bipartite, we have $|V_H| \geq 3$, so by Proposition~\ref{prop:projective}, $H$ admits an edge gadget $F=(V_F,E_F)$ with two specified vertices $a,b \in V_F$. Consider an instance of the $\ListHCol$ problem, that is, a graph $G=(V,E)$ on $n$ vertices and a list $L(v) \subseteq V_H$ for each vertex $v \in V$. We define a transformation that, given such an input, produces a graph $G'=(V',E')$ constructed as follows. We initialize $G'$ as a disjoint union of the graph $G$ and the graph $H=(V_H,E_H)$. Then, for each pair of vertices $v \in V$ and $h \in V_H \setminus L(v)$, we add to the graph a copy $F_{v,h}$ of the graph $F$ where the specified vertices $a$ and $b$ are identified with $v$ and $h$, respectively. Note that each such copy introduces $|V_F|-2$ new vertices to the graph $G'$. This completes the description of the graph $G'$, which is clearly constructible in polynomial time. Observe that the transformation is linear-parameter with respect to the number of vertices, as the number of vertices in $G'$ satisfies
\[|V'| \leq n + |V_H|+(|V_F|-2) \cdot n \cdot |V_H| = O(n).\]

We now verify the correctness of the transformation.
Suppose first that $(G,L)$ is a $\YES$ instance of $\ListHCol$, that is, there exists a homomorphism $f:V \rightarrow V_H$ from $G$ to $H$ satisfying $f(v) \in L(v)$ for all $v \in V$. Define a mapping $f':V' \rightarrow V_H$ as follows. First, for each $v \in V$, set $f'(v) = f(v)$. Since every edge of $G'[V]$ is also an edge of $G$, its endpoints are mapped by $f'$ to adjacent vertices in $H$. Next, for each $h \in V_H$, set $f'(h)=h$. Clearly, the endpoints of every edge in the copy of $H$ in $G'$ are mapped by $f'$ to adjacent vertices in $H$. Finally, for each vertex $v \in V$, it holds that $f(v) \in L(v)$, so for each $h \in V_H \setminus L(v)$, it holds that $f(v) \neq h$, and thus $f'(v) \neq f'(h)$. Since $F$ is an edge gadget for $H$, this ensures that $f'$ can be extended to all the vertices of the copy $F_{v,h}$ of $F$, so that the endpoints of every edge in this copy are mapped to adjacent vertices in $H$. The obtained mapping $f'$ forms a homomorphism from $G'$ to $H$, so $G'$ is $H$-colorable and is thus a $\YES$ instance of $\HCol$.

For the converse direction, suppose that $G'$ is a $\YES$ instance of $\HCol$, that is, there exists a homomorphism $f':V' \rightarrow V_H$ from $G'$ to $H$. The restriction of $f'$ to the copy of $H$ on the vertex set $V_H$ in $G'$ is a homomorphism from $H$ to itself. Since $H$ is a core, there exists an automorphism $\pi:V_H \rightarrow V_H$ of $H$ such that $f'(h) = \pi(h)$ for all $h \in V_H$. Note that the inverse $\pi^{-1}$ of $\pi$ also forms an automorphism of $H$. Consider the mapping $f: V \rightarrow V_H$ defined by $f(v) = \pi^{-1}(f'(v))$ for all $v \in V$. The mapping $f$ forms a homomorphism from $G$ to $H$. Indeed, if $u$ and $v$ are adjacent vertices in $G$, then they are also adjacent in $G'$, so $f'(u)$ and $f'(v)$ are adjacent in $H$, which implies that $f(u)=\pi^{-1}(f'(u))$ and $f(v)=\pi^{-1}(f'(v))$ are adjacent in $H$ as well. We further claim that for every vertex $v \in V$, it holds that $f(v) \in L(v)$. To see this, recall that the graph $G'$ contains, for every $v \in V$ and $h \in V_H \setminus L(v)$, a copy $F_{v,h}$ of $F$ whose specified vertices $a$ and $b$ are identified with $v$ and $h$. The restriction of $f'$ to this copy of $F$ forms a homomorphism to $H$, so the fact that $F$ is an edge gadget for $H$ guarantees that $f'(v) \neq f'(h)$, and thus $f(v) = \pi^{-1}(f'(v)) \neq \pi^{-1}(f'(h)) = h$. Therefore, $f(v) \neq h$ for all $h \in V_H \setminus L(v)$, implying that $f(v) \in L(v)$. We conclude that $(G,L)$ is a $\YES$ instance of $\ListHCol$, as required.
\end{proof}

\subsection{Parameterization by Vertex Cover Number}

Our lower bound on the kernel size of $\HCol$ problems parameterized by the vertex cover number builds on a lower bound on the compressibility of satisfiability problems. A Boolean CNF formula is said to be {\em NAE-satisfiable}, where NAE stands for Not All Equal, if there exists an assignment to its variables such that each clause includes at least one literal evaluated to $\true$ and at least one literal evaluated to $\false$. For a positive integer $q$, a $q$-CNF formula is a CNF formula in which each clause contains exactly $q$ literals. The $\qNAE$ problem asks whether a given $q$-CNF formula is NAE-satisfiable. The following result provides a lower bound on the compression size of this problem in terms of the number of variables.

\begin{theorem}[\cite{JansenK13}]\label{thm:NAE}
For every integer $q \geq 4$ and any real $\eps >0$, the $\qNAE$ problem parameterized by the number of variables $n$ does not admit a compression of size $O(n^{q-1-\eps})$ unless $\NP \subseteq \coNPpoly$.
\end{theorem}
\noindent
As shown in~\cite{JansenK13}, Theorem~\ref{thm:NAE} follows from a lower bound on the compression size of the standard $\qSAT$ problem, proved by Dell and van Melkebeek~\cite{DellM14}, through the simple transformation that appends a single new variable to all clauses. The stated lower bound is known to be essentially optimal~\cite{JansenP17}.

We are ready to prove Theorem~\ref{thm:IntroLowerGenQ(H)}, which asserts that for every non-bipartite projective core graph $H$ with $q(H) \geq 4$ and for any real $\eps >0$, the $\HCol$ problem parameterized by the vertex cover number $k$ does not admit a compression of size $O(k^{q(H)-1-\eps})$ unless $\NP \subseteq \coNPpoly$.

\begin{proof}[ of Theorem~\ref{thm:IntroLowerGenQ(H)}]
For a non-bipartite projective core graph $H=(V_H,E_H)$, set $q = q(H)$, and suppose that $q \geq 4$.
By Theorem~\ref{thm:NAE}, for any real number $\eps >0$, the $\qNAE$ problem parameterized by the number of variables $n$ does not admit a compression of size $O(n^{q-1-\eps})$ unless $\NP \subseteq \coNPpoly$. In what follows, we present a linear-parameter transformation from the $\qNAE$ problem parameterized by the number of variables to the $\HCol$ problem parameterized by the vertex cover number. Therefore, if the $\HCol$ problem parameterized by the vertex cover number $k$ admits a compression of size $O(k^{q-1-\eps})$, then composing this compression with the transformation described below would yield a compression of size $O(n^{q-1-\eps})$ for the $\qNAE$ problem parameterized by the number of variables $n$, implying that $\NP \subseteq \coNPpoly$.

Since $H$ is a non-bipartite projective core, it follows in particular that $|V_H| \geq 3$, so by Proposition~\ref{prop:projective}, $H$ admits an edge gadget $F=(V_F,E_F)$ with two specified vertices $a,b \in V_F$. The definition of the non-adjacency witness number implies that there exists a set $T = \{a_1, \ldots, a_q\} \subseteq V_H$ of size $q$, such that the vertices of $T$ do not share a common neighbor in $H$, while those of each of its $(q-1)$-subsets do. Consider an instance of the $\qNAE$ problem, namely, a $q$-CNF formula $\varphi$ on $n$ variables, denoted $x_1, \ldots, x_n$. We define a transformation that maps such an input to a graph $G=(V,E)$ and a vertex cover $X$ of $G$ as follows. Throughout, indices $j \in [q]$ are interpreted modulo $q$, so $0$ and $1$ are identified with $q$ and $q+1$, respectively. An illustration is given in Figure~\ref{fig:reduction}.
\begin{enumerate}
  \item Initialize the graph $G$ as the disjoint union of the graph $H=(V_H,E_H)$ and a collection of $2qn$ isolated vertices, denoted $t_{i,j}$ and $f_{i,j}$ for $i \in [n]$ and $j \in [q]$.
  \item\label{itm:F} For each of the following pairs of vertices, add to the graph $G$ a copy of the edge gadget $F$ with its specified vertices $a$ and $b$ identified with the vertices of the pair.
      \begin{enumerate}[label=\alph*.]
        \item $(t_{i,j},f_{i,j})$ for all $i \in [n]$ and $j \in [q]$.
        \item $(t_{i,j},t_{i,j+1})$ for all $i \in [n]$ and $j \in [q]$.
        \item $(t_{i,j},h)$ and $(f_{i,j},h)$ for all $i \in [n]$, $j \in [q]$, and $h \in V_H \setminus \{a_j,a_{j+1}\}$.
      \end{enumerate}
      Note that the total number of copies of $F$ added to the graph $G$ is \[2qn+2qn \cdot (|V_H|-2)=2q(|V_H|-1)n,\] and that each of them introduces $|V_F|-2$ new vertices.
  \item For each clause $C = (y_1 \vee \cdots \vee y_q)$ in $\varphi$, add to $G$ a vertex $c$, referred to as a {\em clause vertex}, and connect it to $q$ vertices in $G$ as follows. For each $j \in [q]$, if $y_j = x_i$ for $i \in [n]$, connect $c$ to the vertex $t_{i,j}$, and otherwise, if $y_j = \overline{x_i}$ for $i \in [n]$, connect $c$ to the vertex $f_{i,j}$.
  \item Let $X$ be the set of all vertices of $G$ except the clause vertices, and return the pair $(G,X)$.
\end{enumerate}
Note that the clause vertices are not adjacent in $G$ to one another, so $X$ is a vertex cover of $G$, implying that $(G,X)$ is a valid output. The number of vertices in $X$ satisfies
\[ |X| = |V_H|+2qn+2q(|V_H|-1)n \cdot (|V_F|-2) = O(n),\]
so the transformation is linear-parameter. One can easily check that the transformation can be implemented in polynomial time.

Before the formal correctness proof, let us briefly explain the high-level idea of the construction. The produced graph $G$ includes, for each $i \in [n]$, a gadget that encodes an assignment to the variable $x_i$. This gadget consists of $2q$ vertices, denoted $t_{i,1}, \ldots, t_{i,q}$ and $f_{i,1}, \ldots, f_{i,q}$, together with multiple copies of the edge gadget $F$ that enforce inequality constraints. We consider two mappings from these $2q$ vertices in $G$ to the vertices of $T = \{a_1, \ldots, a_q\}$ in $H$. The first maps each vertex $t_{i,j}$ to $a_j$ and each vertex $f_{i,j}$ to $a_{j+1}$, while the second maps each vertex $t_{i,j}$ to $a_{j+1}$ and each vertex $f_{i,j}$ to $a_j$ (see Figure~\ref{fig:reduction}). We interpret these two mappings as the $\true$ and $\false$ assignments to $x_i$, respectively. A crucial property of the construction, ensured by the copies of the edge gadget $F$, is that every homomorphism from $G$ to $H$ agrees, up to an automorphism of $H$, with one of these two fixed mappings on each gadget. Therefore, the possible assignments to the variables of $\varphi$ are essentially represented by the homomorphisms from the subgraph of $G$ induced by the variable gadgets to $H$. Now, each clause $C$ in $\varphi$ is associated with a vertex $c$ in $G$ connected by $q$ edges to vertices of the variable gadgets, one edge for each literal of $C$. If these literals involve both the $\true$ and $\false$ values under a given assignment, then the corresponding homomorphism from $G$ to $H$ maps the neighbors of $c$ to a proper subset of $T$, allowing $c$ to be assigned a compatible vertex of $H$. Conversely, if all literals of $C$ are evaluated to the same value, the corresponding homomorphism maps the neighbors of $c$ to the entire set $T$, so it cannot be properly extended to $c$. With this intuitive idea in mind, we proceed to the full proof.

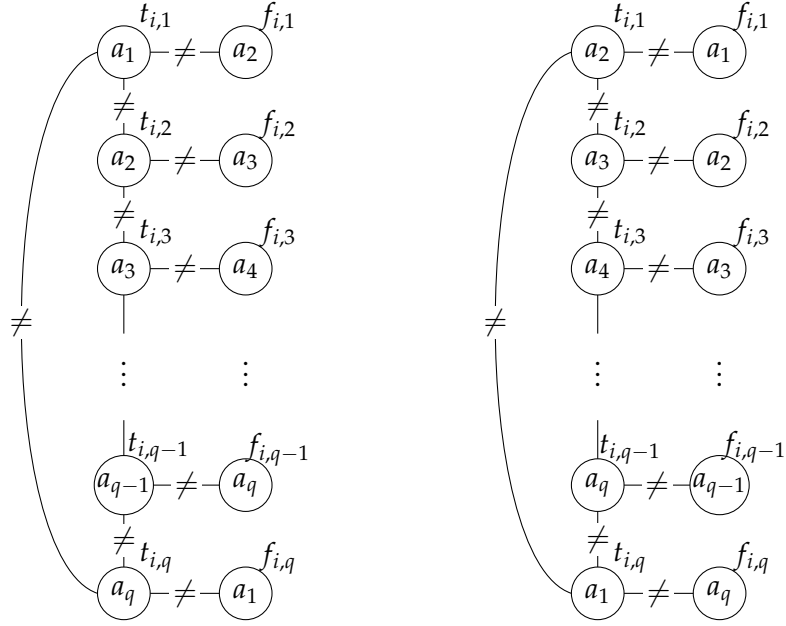
\begin{figure}[!htbp]
\centering 
\begin{tikzpicture}[
    scale=0.90,
    node_number_style/.style={
        circle, draw,
        minimum size=7mm,
        inner sep=0pt,
        outer sep=0pt
    },
    vertical_dots_style/.style={
        rectangle, draw=none, inner sep=0pt,
        minimum height=0.6cm
    }
]
    \pgfmathsetmacro{\vd}{1.6}   
    \pgfmathsetmacro{\hd}{1.8}
    \pgfmathsetmacro{\graphshift}{7}
    \pgfmathsetmacro{\labeloffset}{0.5}
    \pgfmathsetmacro{\xshift}{\labeloffset*cos(45) - 0.15}
    \pgfmathsetmacro{\yshift}{\labeloffset*sin(45) - 0.1}
    \node (L1) at (0,0) [node_number_style] {$a_1$};
    \node at ([shift={(\xshift,\yshift)}]L1.north east) {$t_{i,1}$};
    \node (L2) at (0,-\vd) [node_number_style] {$a_2$};
    \node at ([shift={(\xshift,\yshift)}]L2.north east) {$t_{i,2}$};
    \node (L3) at (0,-2*\vd) [node_number_style] {$a_3$};
    \node at ([shift={(\xshift,\yshift)}]L3.north east) {$t_{i,3}$};
    \node (Lq_minus_1) at (0,-4*\vd) [node_number_style] {$a_{q-1}$};
    \node at ([shift={(\xshift,\yshift)}]Lq_minus_1.north east) {$t_{i,q-1}$};
    \node (Lq) at (0,-5*\vd) [node_number_style] {$a_q$};
    \node at ([shift={(\xshift,\yshift)}]Lq.north east) {$t_{i,q}$};
    \node (R1) at (\hd,0) [node_number_style] {$a_2$};
    \node at ([shift={(\xshift,\yshift)}]R1.north east) {$f_{i,1}$};
    \node (R2) at (\hd,-\vd) [node_number_style] {$a_3$};
    \node at ([shift={(\xshift,\yshift)}]R2.north east) {$f_{i,2}$};
    \node (R3) at (\hd,-2*\vd) [node_number_style] {$a_4$};
    \node at ([shift={(\xshift,\yshift)}]R3.north east) {$f_{i,3}$};
    \node (Rq_minus_1) at (\hd,-4*\vd) [node_number_style] {$a_q$};
    \node at ([shift={(\xshift,\yshift)}]Rq_minus_1.north east) {$f_{i,q-1}$};
    \node (Rq) at (\hd,-5*\vd) [node_number_style] {$a_1$};
    \node at ([shift={(\xshift,\yshift)}]Rq.north east) {$f_{i,q}$};
    \draw (L1) -- (L2) node[midway, fill=white, inner sep=1pt] {$\neq$};
    \draw (L2) -- (L3) node[midway, fill=white, inner sep=1pt] {$\neq$};
    \draw (L3) -- ++(0,-0.6*\vd);  
    \draw (Lq_minus_1) -- ++(0,0.6*\vd); 
    \draw (Lq_minus_1) -- (Lq) node[midway, fill=white, inner sep=1pt] {$\neq$};
    \draw (L1) -- (R1) node[midway, fill=white, inner sep=1pt] {$\neq$};
    \draw (L2) -- (R2) node[midway, fill=white, inner sep=1pt] {$\neq$};
    \draw (L3) -- (R3) node[midway, fill=white, inner sep=1pt] {$\neq$};
    \draw (Lq_minus_1) -- (Rq_minus_1) node[midway, fill=white, inner sep=1pt] {$\neq$};
    \draw (Lq) -- (Rq) node[midway, fill=white, inner sep=1pt] {$\neq$};
    \draw (Lq.west) .. controls +(-1.5cm,0.5cm) and +(-1.5cm,-0.5cm) .. (L1.west)
        node[midway, fill=white, inner sep=1pt] {$\neq$};  
    \node[vertical_dots_style, text depth=0pt] at (0, {-2.9*\vd}) {\vdots};
    \node at (\hd, -2.9*\vd) {\vdots};
    \begin{scope}[xshift=\graphshift cm]
        \node (RL1) at (0,0) [node_number_style] {$a_2$};
        \node at ([shift={(\xshift,\yshift)}]RL1.north east) {$t_{i,1}$};
        \node (RL2) at (0,-\vd) [node_number_style] {$a_3$};
        \node at ([shift={(\xshift,\yshift)}]RL2.north east) {$t_{i,2}$};
        \node (RL3) at (0,-2*\vd) [node_number_style] {$a_4$};
        \node at ([shift={(\xshift,\yshift)}]RL3.north east) {$t_{i,3}$};
        \node (RLq_minus_1) at (0,-4*\vd) [node_number_style] {$a_q$};
        \node at ([shift={(\xshift,\yshift)}]RLq_minus_1.north east) {$t_{i,q-1}$};
        \node (RLq) at (0,-5*\vd) [node_number_style] {$a_1$};
        \node at ([shift={(\xshift,\yshift)}]RLq.north east) {$t_{i,q}$};
        \node (RR1) at (\hd,0) [node_number_style] {$a_1$};
        \node at ([shift={(\xshift,\yshift)}]RR1.north east) {$f_{i,1}$};
        \node (RR2) at (\hd,-\vd) [node_number_style] {$a_2$};
        \node at ([shift={(\xshift,\yshift)}]RR2.north east) {$f_{i,2}$};
        \node (RR3) at (\hd,-2*\vd) [node_number_style] {$a_3$};
        \node at ([shift={(\xshift,\yshift)}]RR3.north east) {$f_{i,3}$};
        \node (RRq_minus_1) at (\hd,-4*\vd) [node_number_style] {$a_{q-1}$};
        \node at ([shift={(\xshift,\yshift)}]RRq_minus_1.north east) {$f_{i,q-1}$};
        \node (RRq) at (\hd,-5*\vd) [node_number_style] {$a_q$};
        \node at ([shift={(\xshift,\yshift)}]RRq.north east) {$f_{i,q}$};
        \draw (RL1) -- (RL2) node[midway, fill=white, inner sep=1pt] {$\neq$};
        \draw (RL2) -- (RL3) node[midway, fill=white, inner sep=1pt] {$\neq$};
        \draw (RL3) -- ++(0,-0.6*\vd); 
        \draw (RLq_minus_1) -- ++(0,0.6*\vd); 
        \draw (RLq_minus_1) -- (RLq) node[midway, fill=white, inner sep=1pt] {$\neq$};
        \draw (RL1) -- (RR1) node[midway, fill=white, inner sep=1pt] {$\neq$};
        \draw (RL2) -- (RR2) node[midway, fill=white, inner sep=1pt] {$\neq$};
        \draw (RL3) -- (RR3) node[midway, fill=white, inner sep=1pt] {$\neq$};
        \draw (RLq_minus_1) -- (RRq_minus_1) node[midway, fill=white, inner sep=1pt] {$\neq$};
        \draw (RLq) -- (RRq) node[midway, fill=white, inner sep=1pt] {$\neq$};
        \draw (RLq.west) .. controls +(-1.5cm,0.5cm) and +(-1.5cm,-0.5cm) .. (RL1.west)
            node[midway, fill=white, inner sep=1pt] {$\neq$}; 
        \node[vertical_dots_style, text depth=0pt] at (0, {-2.9*\vd}) {\vdots};
        \node at (\hd, -2.9*\vd) {\vdots};
    \end{scope}
\end{tikzpicture}
    \caption{The gadget graph for a variable $x_i$. Each $\neq$ sign represents a copy of the edge gadget $F$, which enforces inequality between its endpoints under any homomorphism to $H$. The mapping on the left corresponds to $x_i$ assigned the value $\true$, and the one on the right corresponds to $x_i$ assigned the value $\false$.}
    \label{fig:reduction}
\end{figure}

We show that $\varphi$ is NAE-satisfiable if and only if $G$ is $H$-colorable.
Suppose first that $\varphi$ is NAE-satisfiable, and consider an assignment $\rho:\{x_1, \ldots,x_n\} \rightarrow \{\true,\false\}$ for which both values appear in each clause of $\varphi$. We define a mapping $g:V \rightarrow V_H$ as follows. First, define $g(h) = h$ for all $h \in V_H$. Clearly, the endpoints of every edge of the copy of $H$ in $G$ are mapped by $g$ to adjacent vertices in $H$. Next, for each $i \in [n]$, if $\rho(x_i) = \true$, then set $g(t_{i,j}) = a_j$ and $g(f_{i,j}) = a_{j+1}$ for all $j \in [q]$, and if $\rho(x_i) = \false$, then set $g(t_{i,j}) = a_{j+1}$ and $g(f_{i,j}) = a_{j}$ for all $j \in [q]$. These assignments guarantee that the vertices in each pair from Item~\ref{itm:F} of the construction are mapped by $g$ to distinct vertices of $H$. Since $F$ is an edge gadget for $H$, it follows that $g$ can be extended to all vertices of these copies of $F$, ensuring that the endpoints of each of their edges are mapped to adjacent vertices in $H$.

It remains to extend $g$ to the clause vertices in $G$. Consider any clause $C = (y_1 \vee \cdots \vee y_q)$ in $\varphi$ and its associated vertex $c$. Since the values of the literals in $C$ are not all equal under $\rho$, there exists some $j \in [q]$ for which $y_{j-1}$ is evaluated to $\false$ and $y_j$ is evaluated to $\true$. We claim that two neighbors of $c$ in $G$ are mapped by $g$ to $a_j$. Indeed, if $y_{j-1} = x_i$ for some $i \in [n]$, then $c$ is adjacent to $t_{i,j-1}$ and $\rho(x_i)=\false$, so $g(t_{i,j-1}) =a_j$. Otherwise, $y_{j-1} = \overline{x_i}$ for some $i \in [n]$, so $c$ is adjacent to $f_{i,j-1}$ and $\rho(x_i)=\true$, thus $g(f_{i,j-1}) =a_j$. Similarly, if $y_j = x_i$ for some $i \in [n]$, then $c$ is adjacent to $t_{i,j}$ and $\rho(x_i)=\true$, so $g(t_{i,j}) =a_j$. Otherwise, $y_j = \overline{x_i}$ for some $i \in [n]$, so $c$ is adjacent to $f_{i,j}$ and $\rho(x_i)=\false$, thus $g(f_{i,j}) =a_j$. Since the $q$ neighbors of $c$ are mapped by $g$ to vertices of $T$, we conclude that they are mapped to a proper subset of $T$. The vertices of such a subset share a common neighbor in $H$, so we can define $g(c)$ as a vertex in $H$ that is adjacent to all the images under $g$ of the neighbors of $c$ in $G$. This gives us a homomorphism from $G$ to $H$, implying that $G$ is $H$-colorable.

For the converse direction, suppose that the graph $G$ is $H$-colorable, and consider a homomorphism $g:V \rightarrow V_H$ from $G$ to $H$. The restriction of $g$ to the vertex set $V_H$ of the copy of $H$ in $G$ is a homomorphism from $H$ to itself. Since $H$ is a core, there exists an automorphism $\pi: V_H \rightarrow V_H$ of $H$ such that $g(h)=\pi(h)$ for all $h \in V_H$. Note that its inverse $\pi^{-1}$ is also an automorphism of $H$. Let $g':V \rightarrow V_H$ be the mapping defined by $g'(v) = \pi^{-1}(g(v))$ for all $v \in V$. Note that $g'$ is also a homomorphism from $G$ to $H$, since for every pair of adjacent vertices $u$ and $v$ in $G$, the vertices $g(u)$ and $g(v)$ are adjacent in $H$, and therefore the vertices $g'(u) = \pi^{-1}(g(u))$ and $g'(v) = \pi^{-1}(g(v))$ are adjacent in $H$ as well. Moreover, the mapping $g'$ acts as the identity on $V_H$, because $g'(h) = \pi^{-1}(g(h)) = \pi^{-1}(\pi(h)) = h$ for all $h \in V_H$.

We now define an assignment $\rho:\{x_1, \ldots,x_n\} \rightarrow \{\true,\false\}$ to the variables of $\varphi$. Fix some $i \in [n]$, and consider the copies of the edge gadget $F$ in the graph $G$, as described in Item~\ref{itm:F} of the construction. The restriction of $g'$ to each such copy forms a homomorphism to $H$, hence its specified vertices are mapped by $g'$ to distinct vertices in $H$. This implies that for all $j \in [q]$, the vertices $t_{i,j}$ and $f_{i,j}$ are mapped by $g'$ to distinct vertices, as are $t_{i,j}$ and $t_{i,j+1}$, and since $g'$ acts as the identity on $V_H$, each of $t_{i,j}$ and $f_{i,j}$ is mapped to either $a_j$ or $a_{j+1}$. Therefore, either $g'(t_{i,j})=a_j$ and $g'(f_{i,j})=a_{j+1}$ for all $j \in [q]$, or $g'(t_{i,j})=a_{j+1}$ and $g'(f_{i,j})=a_{j}$ for all $j \in [q]$. We define $\rho(x_i) = \true$ in the former case, and $\rho(x_i) = \false$ in the latter.

We finally show that each clause in $\varphi$ includes two literals with distinct values under $\rho$, implying that $\varphi$ is NAE-satisfiable. Consider a clause $C = (y_1 \vee \cdots \vee y_q)$ in $\varphi$ and its associated vertex $c$. The homomorphism $g'$ maps $c$ to a vertex in $H$ that is adjacent to the images under $g'$ of its $q$ neighbors in $G$. Since each of those neighbors is mapped by $g'$ to a vertex of $T$, and since the vertices of $T$ do not share a neighbor in $H$, it follows that two neighbors of $c$ are mapped by $g'$ to the same vertex $a_j$ with $j \in [q]$. By the above discussion, these neighbors must correspond to the literals $y_{j-1}$ and $y_j$ in $C$. Now, if $y_{j-1} = x_i$ for some $i \in [n]$, then the edge of $c$ corresponding to $y_{j-1}$ connects $c$ and $t_{i,j-1}$, so the fact that $g'(t_{i,j-1})=a_j$ implies that $\rho(x_i)=\false$. Further, if $y_{j-1} = \overline{x_i}$ for some $i \in [n]$, then the edge of $c$ corresponding to $y_{j-1}$ connects $c$ and $f_{i,j-1}$, so the fact that $g'(f_{i,j-1})=a_j$ implies that $\rho(x_{i})=\true$. In both cases, the value of $y_{j-1}$ under $\rho$ is $\false$. Similarly, if $y_j = x_i$ for some $i \in [n]$, then the edge of $c$ corresponding to $y_j$ connects $c$ and $t_{i,j}$, so the fact that $g'(t_{i,j})=a_j$ implies that $\rho(x_i)=\true$. Further, if $y_j = \overline{x_i}$ for some $i \in [n]$, then the edge of $c$ corresponding to $y_j$ connects $c$ and $f_{i,j}$, so the fact that $g'(f_{i,j})=a_j$ implies that $\rho(x_{i  })=\false$. In both cases, the value of $y_j$ under $\rho$ is $\true$. Therefore, the values of the literals in $C$ under $\rho$ are not all equal, as required.
\end{proof}

\section{Non-adjacency Witness Number}\label{sec:NAWN}

In this section, we thoroughly study the non-adjacency witness number of graphs (see Definition~\ref{def:q(G)}). We first provide several useful bounds on this graph quantity, then determine its value for concrete families of graphs, and finally analyze its behavior on random graphs.

Before proceeding, we note that the non-adjacency witness number of a graph can be expressed in terms of the Helly property of its open neighborhoods. A family of sets is said to satisfy the {\em Helly property of order $q$} if every subfamily with an empty intersection that is minimal with respect to containment includes at most $q$ sets. In this language, the non-adjacency witness number $q(G)$ of a graph $G=(V,E)$ is the smallest positive integer $q$ such that the family $\{N_G(v) \mid v \in V\}$ of open neighborhoods in $G$ satisfies the Helly property of order $q$. This notion has been studied in the literature, with particular attention given to graphs $G$ with $q(G)=2$ (see, e.g.,~\cite{DouradoPS06,GroshausLS17}). Such graphs are often called open neighborhood-Helly graphs, and a full characterization of them is given in~\cite{GroshausS07}. An alternative version of this notion replaces open neighborhoods with closed ones. Note that the two notions behave quite differently for certain graph classes. For example, for line graphs, it was proved in~\cite{KanteLMN12} that the family of closed neighborhoods always satisfies the Helly property of order $6$, and yet, their non-adjacency witness number can be arbitrarily large (as demonstrated, for instance, by complete graphs; see Corollary~\ref{cor:q(Km)}).

\subsection{Fundamental Bounds}

The following simple lemma shows that the non-adjacency witness number is sandwiched between the clique number and the maximum degree plus one.

\begin{lemma}\label{lemma:q_bounds}
For every graph $G$, it holds that $\omega(G) \leq q(G) \leq \Delta(G)+1$.
\end{lemma}

\begin{proof}
Let $G=(V,E)$ be a graph. For the left-hand inequality, let $T \subseteq V$ be a clique of size $\omega(G)$ in $G$. By the maximality of $T$, its vertices have no common neighbor, whereas those of every proper subset of $T$ do. Thus, $\omega(G) \leq q(G)$. For the right-hand inequality, let $T \subseteq V$ be a set of vertices with no common neighbor in $G$. To show that there exists a subset $T' \subseteq T$ of size at most $\Delta(G)+1$ with no common neighbor, note that if $|T| \leq \Delta(G)+1$, one may take $T'=T$. Otherwise, let $T'$ be any subset of $T$ of size $\Delta(G) + 1$, and notice that its vertices have no common neighbor in $G$. This implies that $q(G) \leq \Delta(G) +1$, and we are done.
\end{proof}
\noindent
Note that the upper bound in Lemma~\ref{lemma:q_bounds} can also be derived by combining Lemmas~\ref{lemma:q_faith} and~\ref{lemma:faith_Delta+1}.

As an immediate consequence, we determine the non-adjacency witness number of complete graphs.

\begin{corollary}\label{cor:q(Km)}
For every integer $m \geq 1$, it holds that $q(K_m)=m$.
\end{corollary}
\begin{proof}
The complete graph $K_m$ on $m$ vertices clearly satisfies $\omega(K_m)=m$ and $\Delta(K_m)=m-1$, hence it follows from Lemma~\ref{lemma:q_bounds} that $q(K_m) = m$.
\end{proof}

We next show that the non-adjacency witness number of a graph is at least that of its core.

\begin{lemma}\label{lemma:q(core)}
Let $G$ be a graph, and let $G'$ be a core of $G$. Then $q(G') \leq q(G)$.
\end{lemma}

\begin{proof}
Let $G=(V,E)$ be a graph, and let $G'=(V',E')$ be a core of $G$. As a core of a graph is isomorphic to an induced subgraph, we may assume that $G'$ is an induced subgraph of $G$. Let $f: V \rightarrow V'$ be a homomorphism from $G$ to $G'$. The restriction $\pi : V' \rightarrow V'$ of $f$ to $V'$ is a homomorphism from $G'$ to itself, and hence an automorphism of $G'$. Furthermore, the inverse $\pi^{-1}$ of $\pi$ is an automorphism of $G'$ as well. Define the function $g: V \rightarrow V'$ by $g(v) = \pi^{-1}(f(v))$ for all $v \in V$. Observe that $g$ is also a homomorphism from $G$ to $G'$, because for every pair of adjacent vertices $u$ and $v$ in $G$, the vertices $f(u)$ and $f(v)$ are adjacent in $G'$, and so are the vertices $g(u) = \pi^{-1}(f(u))$ and $g(v) = \pi^{-1}(f(v))$. Moreover, the function $g$ maps each vertex of $G'$ to itself, because for every $v \in V'$, we have $g(v) = \pi^{-1}(f(v)) = \pi^{-1}(\pi(v)) = v$.

Now, let $q = q(G)$. To prove that $q(G') \leq q$, consider a set $T \subseteq V'$ of vertices in $G'$ with no common neighbor. We claim that $T$ has no common neighbor in $G$. Indeed, suppose for contradiction that some vertex $v \in V$ is adjacent in $G$ to all vertices of $T$. This implies that $g(v)$ is adjacent in $G'$ to all vertices of $T$, because $g$ is a homomorphism that acts as the identity on $V'$. This contradicts the assumption that $T$ has no common neighbor in $G'$. Finally, given that $T$ has no common neighbor in $G$, there exists a subset $T' \subseteq T$ of size $|T'| \leq q$ with no common neighbor in $G$, and hence with no common neighbor in $G'$. This completes the proof.
\end{proof}

\begin{remark}
The non-adjacency witness number of a graph can be arbitrarily larger than that of its core. For example, for any integer $m \geq 3$, consider the graph $G$ obtained by removing a perfect matching from the complete bipartite graph with $m$ vertices in each part. One can verify that $q(G)=m$, whereas the core $K_2$ of $G$ satisfies $q(K_2)=2$.
\end{remark}

\begin{remark}
Lemma~\ref{lemma:q(core)} shows that replacing a graph with its core does not increase the non-adjacency witness number. This, however, does not hold when replacing it with a general induced subgraph, as demonstrated by the following example. Consider the cube graph $Q_3$ on the vertex set $\{0,1\}^3$, in which two vertices are adjacent if the corresponding vectors differ in exactly one entry. Let $G$ be the graph obtained from $Q_3$ by adding an edge between the vertices $(0,0,0)$ and $(1,1,1)$, and let $G'$ be its induced subgraph on $\{0,1\}^3 \setminus \{(0,0,0)\}$. Observe that $q(G') \geq 3$, as the set $\{(1,0,0),(0,1,0),(0,0,1)\}$ has no common neighbor in $G'$, while each of its $2$-subsets does. Yet, we claim that $q(G) \leq 2$. To see this, consider a set $T$ of vertices in $G$ with no common neighbor. If $T$ includes both $(0,0,0)$ and $(1,1,1)$, then these two already certify that it has no common neighbor. If $T$ includes exactly one of them, say $(0,0,0)$, then the fact that $(1,1,1)$ is not a common neighbor of the vertices in $T$ implies that $T$ includes a vertex whose vector has a single $1$, however, such a vertex shares no neighbor with $(0,0,0)$. Otherwise, $T$ includes neither $(0,0,0)$ nor $(1,1,1)$, so its vertices lie on the cycle of $6$ vertices induced by the remaining vertices of $G$. Since every set of three alternating vertices along this cycle has a common neighbor, it follows that $T$ includes a pair of consecutive vertices in the cycle. Such a pair certifies that $T$ has no common neighbor, so we are done.
\end{remark}

Our next bound on the non-adjacency witness number involves the family of graphs $B_{m,\ell}$, defined as follows.

\begin{definition}\label{def:B_m,l}
For a positive integer $m$ and for an integer $0 \leq \ell \leq m$, let $B_{m,\ell}$ denote the graph on the vertex set $\{u_1, \ldots, u_m\} \cup \{v_1,\ldots, v_\ell\}$, where for every $i \in [\ell]$, the vertex $v_i$ is adjacent to all vertices $u_j$ with $j \in [m] \setminus \{i\}$, and for every $i \in [m] \setminus [\ell]$, the vertex $u_i$ is adjacent to all vertices $u_j$ with $j \in [m] \setminus \{i\}$.
\end{definition}

A few observations are pertinent here.
First, the graph $B_{m,0}$ is isomorphic to the complete graph $K_m$.
Second, for admissible values of $m$ and $\ell$, each of $B_{m+1,\ell}$ and $B_{m+1,\ell+1}$ has a subgraph isomorphic to $B_{m,\ell}$.
And third, in the graph $B_{m,\ell}$, the degree of each vertex $v_i$ with $i \in [\ell]$ is $m-1$, the degree of each vertex $u_i$ with $i \in [\ell]$ is $m-1$, and the degree of each vertex $u_i$ with $i \in [m] \setminus [\ell]$ is $m+\ell-1$. Therefore, the number of edges in $B_{m,\ell}$ is
\begin{eqnarray}\label{eq:E(B_m,l)}
\frac{1}{2} \cdot \Big ( 2\ell \cdot (m-1) +(m-\ell) \cdot (m+\ell-1) \Big ) = \frac{1}{2} \cdot (m^2-\ell^2+2m\ell-m-\ell).
\end{eqnarray}

Equipped with Definition~\ref{def:B_m,l}, we are ready to present the following useful lemma.
\begin{lemma}\label{lemma:q_vs_B}
For a positive integer $q$, let $G$ be a graph that contains no copy (induced or not) of the graph $B_{q,\ell}$ for any $\ell$ with $0 \leq \ell \leq q$. Then $q(G) \leq q-1$.
\end{lemma}

\begin{proof}
Let $G = (V,E)$ be a graph with no copy of $B_{q,\ell}$ for any $\ell$ with $0 \leq \ell \leq q$. We first observe that $G$ contains no copy of $B_{m,\ell}$ for all integers $m$ and $\ell$ satisfying $m \geq q$ and $0 \leq \ell \leq m$. Suppose, toward a contradiction, that $G$ contains a copy of $B_{m,\ell}$ for such integers $m$ and $\ell$. If $\ell \leq q$, then $B_{m,\ell}$ admits a subgraph isomorphic to $B_{q,\ell}$, as follows by applying $m-q$ times the fact that each graph $B_{m'+1,\ell'}$ contains a copy of $B_{m',\ell'}$. However, this contradicts the assumption that $G$ contains no copy of $B_{q,\ell}$ with $0 \leq \ell \leq q$. Otherwise, it holds that $\ell > q$, so $B_{m,\ell}$ admits a subgraph isomorphic to $B_{m-\ell+q,q}$, as follows by applying $\ell-q$ times the fact that each graph $B_{m'+1,\ell'+1}$ contains a copy of $B_{m',\ell'}$. This reduces to the previous case, contradicting again our assumption on $G$.

We are now ready to prove that $q(G) \leq q-1$. Let $T \subseteq V$ be a set of vertices with no common neighbor in $G$, and let $T'$ be a minimal subset of $T$ (with respect to containment) whose vertices have no common neighbor in $G$. Let $m = |T'|$ denote the size of $T'$, and put $T' = \{u_1, \ldots, u_m\}$. Our goal is to show that $m \leq q-1$. By the minimality of $T'$, for every $i \in [m]$, there exists a vertex $v_i \in V$ that is adjacent to all vertices of $T'$ except $u_i$. Note that $v_i$ may or may not be the vertex $u_i$. Let $\ell$ denote the number of indices $i \in [m]$ for which it holds that $u_i \neq v_i$, and notice that $0 \leq \ell \leq m$. It may be assumed, without loss of generality, that $u_i \neq v_i$ for all $i \in [\ell]$, and thus $u_i = v_i$ for all $i \in [m] \setminus [\ell]$. Observe that the subgraph of $G$ induced by the vertices of $\{u_1, \ldots, u_m\} \cup \{v_1,\ldots, v_\ell\}$ contains the graph $B_{m,\ell}$ as a subgraph. In particular, $G$ contains a copy of $B_{m,\ell}$, hence $m \leq q-1$, and we are done.
\end{proof}

As a simple consequence of Lemma~\ref{lemma:q_vs_B}, we relate the non-adjacency witness number of a graph to its degeneracy.

\begin{lemma}\label{lemma:d-deg}
Let $d$ be a positive integer. For every $d$-degenerate graph $G$, it holds that $q(G) \leq d+1$.
\end{lemma}

\begin{proof}
Let $G$ be a $d$-degenerate graph. Observe that the minimum degree of a vertex in the graph $B_{m,\ell}$ is $m-1$. Since every subgraph of $G$ has a vertex of degree at most $d$, it follows that $G$ contains no copy of $B_{d+2,\ell}$ with $0 \leq \ell \leq d+2$. By Lemma~\ref{lemma:q_vs_B}, this implies that $q(G) \leq d+1$, as desired.
\end{proof}

We next consider the non-adjacency witness number of planar graphs. Since every planar graph $G$ is $5$-degenerate, it follows from Lemma~\ref{lemma:d-deg} that it satisfies $q(G) \leq 6$. The following lemma improves on this bound. The proof relies on the well-known consequences of Euler's formula, which state that every (simple) planar graph $G=(V,E)$ on at least three vertices satisfies $|E| \leq 3 |V|-6$, and that if it is triangle-free, then it further satisfies $|E| \leq 2 |V|-4$.

\begin{lemma}\label{lemma:q_planar}
For every planar graph $G$, it holds that $q(G) \leq 4$.
\end{lemma}

\begin{proof}
By Lemma~\ref{lemma:q_vs_B}, it suffices to prove that every planar graph contains no copy of the graph $B_{5,\ell}$ for all $0 \leq \ell \leq 5$. To this end, we show that these graphs are not planar. Let $B_{5,\ell} = (V,E)$, and recall, using~\eqref{eq:E(B_m,l)}, that $|V|=5+\ell$ and $|E|=\frac{1}{2} \cdot (20+9\ell-\ell^2)$. For $0 \leq \ell \leq 3$, it can be easily verified that $|E| > 3|V|-6$, hence the graph $B_{5,\ell}$ is not planar. For $\ell=5$, the graph $B_{5,5}$ is bipartite, and thus triangle-free, and it holds that $|E|>2|V|-4$, so it is not planar as well. Finally, for $\ell=4$, notice that removing from $B_{5,4}$ the four edges that connect the vertex $u_5$ to the vertices $u_i$ with $i \in [4]$ gives a bipartite graph with $9$ vertices and $16$ edges, so it is not planar either. This completes the proof.
\end{proof}
\noindent
Note that the bound provided by Lemma~\ref{lemma:q_planar} is tight, as witnessed by the complete graph $K_4$, which is planar and satisfies $q(K_4)=4$ by Corollary~\ref{cor:q(Km)}.

\subsection{Concrete Families}

We turn to determining the non-adjacency witness number of Kneser graphs (see Section~\ref{sec:graphs}). Lov{\'{a}}sz~\cite{LovaszKneser} proved that, for positive integers $m$ and $r$ with $m \geq 2r$, the chromatic number of the Kneser graph $K(m,r)$ is $m-2r+2$. The latter is known to coincide with the minimum possible dimension of an orthogonal representation of $K(m,r)$ over any field~\cite{Haviv18topo,AlishahiM21}. Here we show that it also coincides with the non-adjacency witness number.

\begin{lemma}\label{lemma:q(K(m,r))}
For all positive integers $m$ and $r$ with $m \geq 2r$, it holds that $q(K(m,r))=m-2r+2$.
\end{lemma}

\begin{proof}
Fix positive integers $m$ and $r$ with $m \geq 2r$. Consider the $m-2r+2$ vertices in $K(m,r)$ of the form $[r-1] \cup \{i\}$ for $i \in \{r, \ldots,m-r+1\}$. Observe that the union of these vertices includes $m-r+1$ elements, implying that no $r$-subset of $[m]$ is disjoint from all of them, hence they do not have a common neighbor. However, if one of these vertices is excluded, say $[r-1] \cup \{i\}$, then all the remaining vertices are adjacent to the vertex $([m] \setminus [m-r+1]) \cup \{i\} $. It follows that $q(K(m,r)) \geq m-2r+2$.

On the other hand, let $T$ be a collection of vertices with no common neighbor in $K(m,r)$. Since no $r$-subset of $[m]$ is disjoint from all the vertices of $T$, it follows that the union of the members of $T$ includes at least $m-r+1$ elements of $[m]$. We define a collection of vertices $T' \subseteq T$ as follows. Start with $T' = \{A\}$ for an arbitrary vertex $A \in T$, and as long as the total number of elements in the vertices of $T'$ is smaller than $m-r+1$, add to $T'$ a set from $T$ that includes an element not yet covered by any set in $T'$. This is possible, because the sets in $T$ include at least $m-r+1$ distinct elements. The first chosen set $A$ covers $r$ elements, and each additional set added to $T'$ contributes at least one new element, hence the collection $T'$ obtained in the process includes at most $m-2r+2$ sets. Since the union of the members of $T'$ includes at least $m-r+1$ elements, it follows that the vertices in $T'$ have no common neighbor. This, in turn, shows that $q(K(m,r)) \leq m-2r+2$, implying that $q(K(m,r))=m-2r+2$, as required.
\end{proof}

We further identify the non-adjacency witness number of cycles.
\begin{lemma}\label{lemma:q(cycle)}
For an integer $m \geq 3$, the value of $q(C_m)$ is $3$ if $m \in \{3,6\}$ and is $2$ otherwise.
\end{lemma}

\begin{proof}
For an integer $m \geq 3$, consider the cycle $C_m$ on $m$ vertices.
Note that $C_3$ is the complete graph on three vertices, so Corollary~\ref{cor:q(Km)} implies that $q(C_3)=3$.
For any $m \geq 4$, we have $\Delta(C_m)=2$ and $\omega(C_m)=2$, so Lemma~\ref{lemma:q_bounds} implies that $2 \leq q(C_m) \leq 3$. For $m = 6$, any set of three alternating vertices along the cycle has no common neighbor, while each of its $2$-subsets does. This shows that $q(C_6) \geq 3$, and thus $q(C_6)=3$. For $m \notin \{3,6\}$, the cycle $C_m$ contains no copy of $K_3$ nor $C_6$, hence it contains no copy of the graph $B_{3,\ell}$ for any $\ell$ with $0 \leq \ell \leq 3$. Lemma~\ref{lemma:q_vs_B} then gives $q(C_m) \leq 2$, implying that $q(C_m) = 2$.
\end{proof}

\subsection{Behavior on Random Graphs}

For a positive integer $n$, let $G(n,1/2)$ denote the random graph on a fixed set of $n$ labeled vertices, where each possible edge is included independently with probability $1/2$. If the graph $G(n,1/2)$ satisfies a given property with probability that tends to $1$ as $n$ tends to infinity, we say that the property holds {\em asymptotically almost surely}. We now turn to determining the typical non-adjacency witness number of $G(n,1/2)$.

\begin{theorem}\label{thm:G(n,1/2)}
The random graph $G = G(n,1/2)$ satisfies $q(G) \leq 2 \cdot \log n$ asymptotically almost surely.
\end{theorem}

\begin{proof}
Let $G = G(n,1/2)$ be a random graph on $n$ vertices, and set $m = 2 \cdot \log n + 1$.
For an integer $0 \leq \ell \leq m$, let $X_{\ell}$ denote the random variable representing the number of copies of the graph $B_{m,\ell}$ in $G$ (see Definition~\ref{def:B_m,l}). We claim that its expectation satisfies
\[ \Expec{}{X_{\ell}} \leq  \binom{n}{m+\ell} \cdot \binom{m+\ell}{\ell} \cdot \binom{m}{\ell} \cdot \ell ! \cdot 2^{-(m^2-\ell^2+2m\ell-m-\ell)/2}.\]
Indeed, there are $\binom{n}{m+\ell}$ possible ways to choose $m+\ell$ vertices that participate in a copy of $B_{m,\ell}$, $\binom{m+\ell}{\ell}$ ways to divide them into two sets $\{u_1,\ldots, u_m\}$ and $\{v_1,\ldots, v_\ell\}$, and $\binom{m}{\ell} \cdot \ell !$ ways to choose from the first set the $\ell$ vertices identified as $u_1, \ldots, u_\ell$ and order them in a row. Then, the probability that this choice corresponds to a copy of $B_{m,\ell}$ is $1/2$ raised to the power of the number of edges in $B_{m,\ell}$, which is given in~\eqref{eq:E(B_m,l)}. This yields the desired bound on $\Expec{}{X_{\ell}}$. Using the inequality $\binom{n}{k} \leq \frac{n^k}{k!}$, we further obtain that
\begin{eqnarray*}
\Expec{}{X_{\ell}} & \leq & \frac{n^{m+\ell}}{(m+\ell)!} \cdot \frac{(m+\ell)!}{\ell ! \cdot m !} \cdot \frac{m!}{\ell ! \cdot (m-\ell)!} \cdot \ell ! \cdot 2^{-(m^2-\ell^2+2m\ell-m-\ell)/2} \\
& = & \frac{1}{\ell ! \cdot (m-\ell)!} \cdot 2^{(m+\ell)\cdot \log n -(m^2-\ell^2+2m\ell-m-\ell)/2}.
\end{eqnarray*}
Note that the exponent of $2$ in the above expression is a convex quadratic function of $\ell$, so it attains its maximum over the interval $[0,m]$ at one of the endpoints. For our choice of $m$, it vanishes at both endpoints, hence for every $0 \leq \ell \leq m$, it holds that
\[ \Expec{}{X_{\ell}} \leq \frac{1}{\ell ! \cdot (m-\ell)!} \leq \frac{1}{(\log n)!},\]
where the second inequality follows from the fact that either $\ell$ or $m-\ell$ exceeds $\log n$.

Finally, let $X = \sum_{\ell=0}^{m}{X_\ell}$ denote the random variable representing the total number of copies of the graphs $B_{m,\ell}$ in $G$ for $0 \leq \ell \leq m$. By linearity of expectation, it follows that $\Expec{}{X} \leq \frac{m+1}{(\log n) !}$, so $\Expec{}{X}$ tends to $0$ as $n$ tends to infinity. By Markov's inequality, the probability of the event $X \geq 1$ also tends to $0$ as $n$ tends to infinity. Therefore, the probability that there is no copy of $B_{m,\ell}$ in $G$ for any $\ell$ with $0 \leq \ell\leq m$ tends to $1$. By Lemma~\ref{lemma:q_vs_B}, this implies that $q(G) \leq m-1 = 2 \cdot \log n$ asymptotically almost surely, as desired.
\end{proof}

\begin{remark}
The bound provided by Theorem~\ref{thm:G(n,1/2)} is essentially tight.
Indeed, it is well known that the random graph $G=G(n,1/2)$ satisfies $\omega(G) \geq (2-o(1)) \cdot \log n$ asymptotically almost surely, with the $o(1)$ term tending to $0$ as $n$ tends to infinity (see, e.g.,~\cite[Chapter~7.2]{FriezeKBook}). By Lemma~\ref{lemma:q_bounds}, this implies that it satisfies $q(G) \geq (2-o(1)) \cdot \log n$ asymptotically almost surely as well.
\end{remark}

\section{Kernel Bounds for Concrete Target Graphs}\label{sec:app}

In this section, we present illustrative applications of our upper and lower bounds on the kernel size of $\HCol$ problems parameterized by the vertex cover number for various target graphs $H$.

\subsection{The Combinatorial Kernel}

We begin with upper bounds that follow from the combinatorial kernel given in Theorem~\ref{thm:Intro_kernel_q}, combined with our analysis of the non-adjacency witness number from Section~\ref{sec:NAWN}.

\begin{theorem}\label{thm:applications_q}
For each of the following cases, the $\HCol$ problem parameterized by the vertex cover number $k$ admits a kernel with $O(k^q)$ vertices and bit-size $O(k^q)$.
\begin{enumerate}
  \item $H$ is $d$-degenerate for an integer $d$, and $q=d+1$.
  \item $H$ is planar, and $q=4$.
  \item\label{itm:cycle} $H = C_{2m+1}$ for an integer $m \geq 2$, and $q=2$.
\end{enumerate}
\end{theorem}

\begin{proof}
The three items of the theorem immediately follow from Theorem~\ref{thm:Intro_kernel_q}, when combined, respectively, with Lemmas~\ref{lemma:d-deg},~\ref{lemma:q_planar}, and~\ref{lemma:q(cycle)}.
\end{proof}

\noindent
Recall that the kernel implied by~\cite{JansenP19color} for the $\HCol$ problem parameterized by the vertex cover number $k$ has $O(k^q)$ vertices and bit-size $O(k^q \cdot \log k)$, where $q = \Delta(H)$. Yet, Theorem~\ref{thm:applications_q} provides kernel sizes bounded by a fixed polynomial for the class of $d$-degenerate graphs with a fixed integer $d$ and for the class of planar graphs, despite their unbounded maximum degree. The third item of Theorem~\ref{thm:applications_q} essentially follows from~\cite{JansenP19color}, though with an additional logarithmic factor in the bit-size, and is included here primarily for completeness.

We next state a result concerning a typical target graph $H$. It follows by combining Theorem~\ref{thm:Intro_kernel_q} with Theorem~\ref{thm:G(n,1/2)}.

\begin{theorem}\label{thm:kernel_random}
Let $H = G(h,1/2)$ be a random graph on $h$ vertices. Then, with probability tending to $1$ as $h$ tends to infinity, the $\HCol$ problem parameterized by the vertex cover number $k$ admits a kernel with $O(k^q)$ vertices and bit-size $O(k^q)$ for $q = 2 \cdot \log h$.
\end{theorem}
\noindent
It is well known that the maximum degree in a random graph $H = G(h,1/2)$ on $h$ vertices is, asymptotically almost surely, at least $h \cdot (1/2-o(1))$ (see, e.g.,~\cite[Chapter ~3.2]{FriezeKBook}). Therefore, for a typical graph $H$, Theorem~\ref{thm:kernel_random} provides a kernel for the $\HCol$ problem parameterized by the vertex cover number, where the degree of the polynomial that bounds the kernel size is exponentially smaller than that of the kernel from~\cite{JansenP19color}.

\subsection{The Algebraic Kernel}

We proceed with applications of our algebraic kernel from Theorem~\ref{thm:kernel_IR}. The first concerns the kernel complexity of problems associated with the dimension of orthogonal graph representations over finite fields. Recall that for an integer $d$ and a field $\Fset$, the $d\dODP_\Fset$ problem asks, given an input graph $G$, whether $G$ admits a $d$-dimensional orthogonal representation over $\Fset$ (see Definition~\ref{def:OR}). When parameterized by the vertex cover number, the input also includes a vertex cover of $G$, whose size $k$ is the parameter of the problem. We prove Theorem~\ref{thm:IntroOD}, which guarantees a kernel with $O(k^{d-1})$ vertices and bit-size $O(k^{d-1} \cdot \log k)$ for every $d \geq 3$ and every finite field $\Fset$. This nearly matches a lower bound from~\cite{HavivR24}.

\begin{proof}[ of Theorem~\ref{thm:IntroOD}]
Fix an integer $d \geq 3$ and a finite field $\Fset$. Let $H = H(\Fset,d)$ denote the graph whose vertices are all the non-self-orthogonal vectors in $\Fset^d$, where two such vectors $x,y \in \Fset^d$ are adjacent if and only if they satisfy $\langle x,y \rangle = 0$ over $\Fset$. Observe that a graph is $H$-colorable precisely when it admits a $d$-dimensional orthogonal representation over $\Fset$. Therefore, the $\HCol$ problem coincides with the $d\dODP_\Fset$ problem. The assignment that maps each vertex of $H$ to its associated vector in $\Fset^d$ forms a faithful $d$-dimensional orthogonal representation over $\Fset$, so, in particular, by Remark~\ref{remark:OD_ID}, the graph $H$ has a faithful $d$-dimensional independent representation over $\Fset$. Note that the field $\Fset$ is efficient, as are all finite fields, enabling us to apply Theorem~\ref{thm:kernel_IR} with the graph $H$ and the field $\Fset$, and thereby obtain a kernel for the $\HCol$ problem parameterized by the vertex cover number $k$ with $O(k^{d-1})$ vertices and bit-size $O(k^{d-1} \cdot \log k)$. The proof is now complete.
\end{proof}

The next application deals with the $\HCol$ problem parameterized by the vertex cover number $k$, where $H$ is a Kneser graph $K(m,r)$. Since the maximum degree in $K(m,r)$ is $\binom{m-r}{r}$, it follows from~\cite{JansenP19color} that the problem admits a kernel with $O(k^{\binom{m-r}{r}})$ vertices. A notable improvement is obtained by combining Theorem~\ref{thm:Intro_kernel_q} with Lemma~\ref{lemma:q(K(m,r))}, yielding a kernel with $O(k^{m-2r+2})$ vertices. The following theorem improves this by another linear factor.

\begin{theorem}\label{thm:kernel_K(m,r)}
For all positive integers $m$ and $r$ with $m > 2r$, the $K(m,r)$-$\Col$ problem parameterized by the vertex cover number $k$ admits a kernel with $O(k^{m-2r+1})$ vertices and bit-size $O(k^{m-2r+1} \cdot \log k)$.
\end{theorem}

Theorem~\ref{thm:kernel_K(m,r)} follows directly from Theorem~\ref{thm:kernel_IR} combined with the following result.

\begin{theorem}\label{thm:IR_kneser}
For all positive integers $m$ and $r$ with $m \geq 2r$, there exists an integer $c=c(m,r)$ such that for every field $\Fset$ with $|\Fset| \geq c$, the graph $K(m,r)$ has a faithful $(m-2r+2)$-dimensional independent representation over $\Fset$.
\end{theorem}

The proof of Theorem~\ref{thm:IR_kneser} uses the following `general position' result.

\begin{proposition}[{\cite[Theorem~3.13]{BabaiF92}}]\label{prop:general_pos}
For positive integers $n,m,t$ with $m \geq t$, let $\Fset$ be a field with $|\Fset| > (m-t) \cdot (n+1)$, and let $U_1, \ldots, U_n$ be a collection of $n$ subspaces of $\Fset^m$. Then there exists a linear transformation $\varphi:\Fset^m \rightarrow \Fset^t$, such that for all $i \in [n]$, it holds that $\dim(\varphi(U_i)) = \min(\dim(U_i),t)$.
\end{proposition}

We are ready to prove Theorem~\ref{thm:IR_kneser}.

\begin{proof}[ of Theorem~\ref{thm:IR_kneser}]
Fix positive integers $m$ and $r$ with $m \geq 2r$ and a sufficiently large field $\Fset$.
Let $M \in \Fset^{(r-1) \times m}$ be a matrix, such that every set of $r-1$ of its columns is linearly independent. Such a matrix can be obtained, for instance, via a Vandermonde-type construction, namely, choose $m$ distinct field elements $\alpha_1, \ldots, \alpha_m \in \Fset$, and for each $i \in [m]$, set the $i$th column of $M$ to be $(1,\alpha_i,\ldots,\alpha_i^{r-2})^t$.

For every vertex $A$ in $K(m,r)$, let $x_A \in \Fset^m$ be a vector satisfying $M x_A = 0$, such that the support of $x_A$ is exactly $A$, meaning that the entries of $x_A$ corresponding to indices in $A$ are nonzero, and all others are zeros. To see that such a vector exists, fix the entries of $x_A$ with indices outside $A$ to be zero. Then, to determine the remaining $r$ entries of $x_A$, consider the $(r-1) \times r$ submatrix $M_A$ of $M$ formed by the columns indexed by $A$, and solve the system $M_Ay = 0$ for $y \in \Fset^{r}$. Since every set of $r-1$ columns of $M$ is linearly independent, fixing one entry of $y$ to a nonzero value yields a unique solution for $y$. Moreover, all entries of $y$ are nonzero, as otherwise, we would obtain $r-1$ columns of $M$ that are linearly dependent.

For every vertex $B$ in $K(m,r)$, let $U_B$ denote the subspace of $\Fset^m$ spanned by the vectors corresponding to the neighbors of $B$, that is, $U_B = \linspan(\{x_C \mid C \in N_{K(m,r)}(B)\})$. Further, for every pair of vertices $A$ and $B$, let $U^A_B$ denote the subspace of $\Fset^m$ spanned by the vector $x_A$ together with the vectors associated with the neighbors of $B$, that is, $U^A_B = \linspan(\{x_A\} \cup \{x_C \mid C \in N_{K(m,r)}(B)\})$. We observe that for every pair of non-adjacent vertices $A$ and $B$ in the graph $K(m,r)$, it holds that $\dim(U_B^A) = \dim(U_B)+1$. Indeed, if $A$ and $B$ are not adjacent, then they are not disjoint, so there exists some element $i \in A \cap B$. It follows that the $i$th entry of $x_A$ is nonzero, whereas the $i$th entry in every vector of $U_B$ is zero. This implies that $x_A \notin U_B$, and consequently, $\dim(U_B^A) = \dim(U_B)+1$, as desired.

We further observe that for every vertex $B$ in $K(m,r)$, it holds that $\dim(U_B) \leq m-2r+1$. To see this, note that the subspace $U_B$ is supported on the $m-r$ entries with indices in $[m] \setminus B$. Moreover, every vector $x \in U_B$ satisfies $Mx=0$. By $m \geq 2r$, the restriction of $M$ to the $m-r$ columns with indices in $[m] \setminus B$ contains an $(r-1) \times (r-1)$ submatrix of $M$. Since the columns of such a submatrix are linearly independent, so are its rows. It follows that the vectors of $U_B$ are supported on $m-r$ entries and satisfy a system of $r-1$ linearly independent equations on this support, implying that $\dim(U_B) \leq (m-r)-(r-1) = m-2r+1$, as required. This in turn yields that for every pair of non-adjacent vertices $A$ and $B$, $\dim(U^A_B) = \dim(U_B)+1 \leq m-2r+2$.

We finally apply Proposition~\ref{prop:general_pos} with $t=m-2r+2$ and with the subspaces $U_B$ for all vertices $B$ as well as the subspaces $U^A_B$ for all pairs of non-adjacent vertices $A,B$. Assuming that $\Fset$ is sufficiently large, and noting that the dimensions of all these subspaces do not exceed $t$, this yields a linear transformation $\varphi:\Fset^m \rightarrow \Fset^{m-2r+2}$ that preserves their dimensions. In particular, for every pair of non-adjacent vertices $A,B$, it holds that $\dim(\varphi(U^A_B)) = \dim(\varphi(U_B))+1$, which implies that $\varphi(x_A) \notin \varphi(U_B)$. Therefore, the assignment of the vector $\varphi(x_A) \in \Fset^{m-2r+2}$ to each vertex $A$ in $K(m,r)$ forms a faithful $(m-2r+2)$-dimensional independent representation over $\Fset$, thus completing the proof.
\end{proof}

\begin{remark}
The faithful independent representation of $K(m,r)$ provided by Theorem~\ref{thm:IR_kneser} attains the minimum possible dimension. Indeed, by Lemma~\ref{lemma:q_faith}, the dimension of any such representation is bounded from below by the non-adjacency witness number of $K(m,r)$, which, by Lemma~\ref{lemma:q(K(m,r))}, equals $m-2r+2$.
\end{remark}

Note that when $H=K(2r+1,r)$ for an integer $r \geq 1$,  Theorem~\ref{thm:kernel_K(m,r)} provides a kernel with $O(k^2)$ vertices and bit-size $O(k^2 \cdot \log k)$. This in particular holds for the Petersen graph $K(5,2)$, which is presented in Figure~\ref{fig:P} along with a faithful $3$-dimensional orthogonal (and thus independent) representation over $\R$.

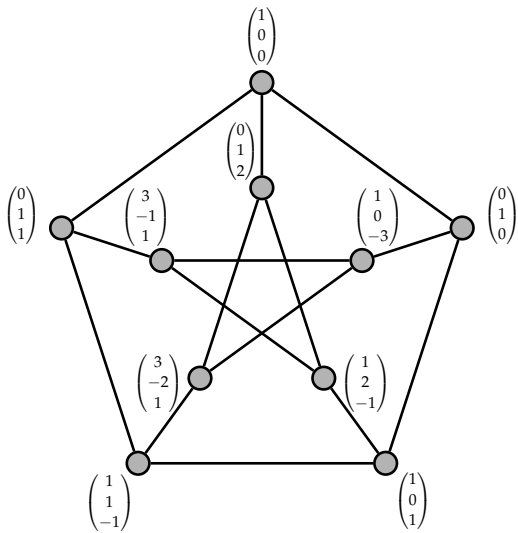
\begin{figure}[htbp]
    \centering
    \begin{tikzpicture}[
      vertex/.style={circle, draw=black, fill=gray!60, line width=1pt, minimum size=3mm, inner sep=0pt},
      vector/.style={draw=none, fill=none, font=\tiny},
      edge/.style={line width=1pt},
      scale=2.8
    ]

    \foreach \i in {0,...,4} {
      \node[vertex] (O\i) at ({90-72*\i}:1) {};
    }

    \foreach \i/\j in {0/5,1/6,2/7,3/8,4/9} {
      \node[vertex] (I\j) at ({90-72*\i}:0.5) {};
    }

    \foreach \i/\j in {0/1,1/2,2/3,3/4,4/0} {
      \draw[edge] (O\i) -- (O\j);
    }
    \foreach \i/\j in {5/7,7/9,9/6,6/8,8/5} {
      \draw[edge] (I\i) -- (I\j);
    }
    \foreach \i/\j in {0/5,1/6,2/7,3/8,4/9} {
      \draw[edge] (O\i) -- (I\j);
    }

    \node[vector] at ($(O0)+(90:0.22)$)   {$\begin{pmatrix} 1 \\ 0 \\ 0 \end{pmatrix}$};
    \node[vector] at ($(O1)+(18:0.2)$)   {$\begin{pmatrix} 0 \\ 1 \\ 0 \end{pmatrix}$};
    \node[vector] at ($(O2)+(-54:0.22)$)  {$\begin{pmatrix} 1 \\ 0 \\ 1 \end{pmatrix}$};
    \node[vector] at ($(O3)+(-126:0.23)$) {$\begin{pmatrix} 1 \\ 1 \\ -1 \end{pmatrix}$};
    \node[vector] at ($(O4)+(162:0.2)$)  {$\begin{pmatrix} 0 \\ 1 \\ 1 \end{pmatrix}$};

    \node[vector] at ($(I5)+(120:0.2)$)  {$\begin{pmatrix} 0 \\ 1 \\ 2 \end{pmatrix}$};
    \node[vector] at ($(I6)+(70:0.22)$)   {$\begin{pmatrix} 1 \\ 0 \\ -3 \end{pmatrix}$};
    \node[vector] at ($(I7)+(-7:0.2)$)   {$\begin{pmatrix} 1 \\ 2 \\ -1 \end{pmatrix}$};
    \node[vector] at ($(I8)+(187:0.2)$)  {$\begin{pmatrix} 3 \\ -2 \\ 1 \end{pmatrix}$};
    \node[vector] at ($(I9)+(110:0.22)$)  {$\begin{pmatrix} 3 \\ -1 \\ 1 \end{pmatrix}$};

    \end{tikzpicture}
    \caption{The Petersen graph --- A faithful $3$-dimensional orthogonal representation over $\R$.}
    \label{fig:P}
\end{figure}

\subsection{Lower Bounds}\label{sec:lowerConcrete}

We begin with applications of Theorem~\ref{thm:IntroLowerGen}, which supplies an almost quadratic lower bound on the compressibility of $\HCol$ problems parameterized by the number of vertices, whenever $H$ is a non-bipartite projective core graph. In~\cite{LaroseT02}, many classical graphs were shown to be projective, including complete graphs on at least three vertices, odd cycles, and Kneser graphs (see also~\cite{Larose02}). It is straightforward to verify that complete graphs on at least three vertices, as well as odd cycles, are non-bipartite cores. It is also known that for all positive integers $m$ and $r$ with $m > 2r$, the Kneser graph $K(m,r)$ is a core (see, e.g.,~\cite[Theorem~7.9.1]{GodsilRbook}), and its chromatic number is $m-2r+2 >2$, so it is not bipartite. This allows us to deduce the following result from Theorem~\ref{thm:IntroLowerGen}. 

\begin{theorem}\label{thm:LowerCycleKneser}
For each of the following cases, for any real $\eps >0$, the $\HCol$ problem parameterized by the number of vertices $n$ does not admit a compression of size $O(n^{2-\eps})$ unless $\NP \subseteq \coNPpoly$.
\begin{enumerate}
  \item\label{itm:complete} $H = K_{m}$ for an integer $m \geq 3$.
  \item\label{itm:cycleL} $H = C_{2m+1}$ for an integer $m \geq 1$.
  \item\label{itm:Kneser} $H = K(m,r)$ for integers $m$ and $r$ with $m > 2r$ and $r \geq 1$.
\end{enumerate}
\end{theorem}

We turn to a lower bound on the compressibility of the $\HCol$ problem for a typical graph $H$. It follows from~\cite{HellN92,LuczakN04} that a random graph $H=G(h,1/2)$ on $h$ vertices is asymptotically almost surely a projective core (see also~\cite{OkrasaR21,PiecykR21}). Moreover, such a graph is asymptotically almost surely non-bipartite, as indicated, for example, by its typical clique number. Combining these facts with Theorem~\ref{thm:IntroLowerGen}, we obtain the following result.

\begin{theorem}
Let $H = G(h,1/2)$ be a random graph on $h$ vertices. Then, with probability tending to $1$ as $h$ tends to infinity, for any real $\eps >0$, the $\HCol$ problem parameterized by the number of vertices $n$ does not admit a compression of size $O(n^{2-\eps})$ unless $\NP \subseteq \coNPpoly$.
\end{theorem}

To wrap up this section, we recall that a lower bound on the compression size of a graph problem when parameterized by the number of vertices also yields a lower bound on its kernel complexity under the vertex cover number parameterization, as the entire vertex set of a graph always constitutes a vertex cover. For certain graphs $H$, this gives a near-optimal lower bound on the kernel complexity of the $\HCol$ problem parameterized by the vertex cover number. One such case is when $H$ is an odd cycle, for which the upper bound from Item~\ref{itm:cycle} of Theorem~\ref{thm:applications_q} closely matches the lower bound implied by Item~\ref{itm:cycleL} of Theorem~\ref{thm:LowerCycleKneser}. Another case occurs when $H = K(2r+1,r)$ for an integer $r \geq 1$, as established by Theorem~\ref{thm:kernel_K(m,r)} and by Item~\ref{itm:Kneser} of Theorem~\ref{thm:LowerCycleKneser}. This completes the proof of Theorem~\ref{thm:IntroLower}. The following theorem extends this lower bound to all non-bipartite Kneser graphs, essentially matching the upper bound given in Theorem~\ref{thm:kernel_K(m,r)}.

\begin{theorem}\label{thm:LowerKneser}
For all positive integers $m$ and $r$ with $m>2r$ and for any real $\eps >0$, the $K(m,r)$-$\Col$ problem parameterized by the vertex cover number $k$ does not admit a compression of size $O(k^{m-2r+1-\eps})$ unless $\NP \subseteq \coNPpoly$.
\end{theorem}

\begin{proof}
As noted at the beginning of the section, for positive integers $m$ and $r$ with $m>2r$, the graph $K(m,r)$ is a non-bipartite projective core. By Lemma~\ref{lemma:q(K(m,r))}, we have $q(K(m,r))=m-2r+2$. The statement of the theorem thus follows for $m \geq 2r+2$ from Theorem~\ref{thm:IntroLowerGenQ(H)}. The remaining case of $m=2r+1$ already follows from Item~\ref{itm:Kneser} of Theorem~\ref{thm:LowerCycleKneser}, so we are done.
\end{proof}

\section{Concluding Remarks}\label{sec:conclude}

This paper studies the kernelization complexity of $\HCol$ problems parameterized by the vertex cover number and presents two kernels for them. The first is purely combinatorial, with size governed by the non-adjacency witness number of $H$, while the second is more algebraic in nature, relying on the existence of a low-dimensional faithful independent representation of $H$ over some field. We have shown that the combinatorial kernel, despite its simplicity, achieves significant improvements over prior work for certain graphs $H$, and that the algebraic kernel can sometimes save a further linear factor in the kernel size. Moreover, we have provided lower bounds on the kernel size for projective graphs $H$ and have demonstrated that in certain cases our kernels attain near-optimal bounds.

The natural challenge raised by this work is to determine, for general graphs $H$, the exact polynomial behavior of the kernel complexity of the $\HCol$ problem parameterized by the vertex cover number. In particular, it would be valuable to obtain lower bounds on the kernel complexity of the problem for non-bipartite core graphs $H$ that are not projective. For such graphs $H$, it would also be interesting to decide whether the $\HCol$ problem admits no sub-quadratic compression when parameterized by the number of vertices. These are the remaining unresolved instances of a question raised in~\cite{ChenJOPR23}. Finally, it would be of interest to explore the kernel complexity of $\HCol$ problems under alternative parameterizations, a direction that has already seen several developments, e.g., in~\cite{JansenK13,JansenP19color,HavivR24}.

\bibliographystyle{abbrv}
\bibliography{Hcol}

\end{document}